\documentclass[aps,prx,reprint,nofootinbib,longbibliography]{revtex4-2}

\usepackage{amsmath,amssymb,bm,mathtools}
\usepackage{amsthm}
\usepackage{booktabs}
\usepackage{array}
\usepackage{hyperref}
\usepackage{CJKutf8}
\newcommand{\ChineseName}{%
  \begin{CJK}{UTF8}{gbsn}尹卫国\end{CJK}%
}

\newcommand{\ignore}[1]{}

\newtheorem{theorem}{Theorem}

\begin{document}

\title{A Human--AI Theorem Connecting Spontaneous and Field-Induced Mechanisms of Collective Behavior in One Dimension}
\author{Weiguo Yin (\ChineseName)}
\email{wyin@bnl.gov}
\affiliation{Condensed Matter Physics and Materials Science Division, Brookhaven National Laboratory, Upton, New York 11973, USA}

\date{\today}

\begin{abstract}
Can an artificial intelligence (AI) generate a scientific hypothesis outside a human collaborator's active hypothesis space (AHS), and can human--AI research be organized to make such breakthroughs more likely? We document such a case while proving a theorem that connects two basic organizing mechanisms of statistical physics: collective behavior arising in zero field from competing interactions and that induced or controlled by an external field. A zero-field $O(n)$-vector open chain with arbitrary inhomogeneous nearest- and next-nearest-neighbor interaction functions $U_i(\bm {S}_i\cdot{\bm S}_{i+1})$ and $V_i({\bm S}_i\cdot{\bm S}_{i+2})$ is microscopically, via a temperature-independent mapping at the Hamiltonian level, equivalent to a simpler $O(n)$ open chain with nearest-neighbor interaction $V_i(\boldsymbol\sigma_i\cdot\boldsymbol\sigma_{i+1})$ and axial single-spin potential $U_i(\sigma_i^z)$ for every integer $n\ge1$ and every system size $L\ge1$. The homogeneous linear specialization maps the foundational frustrated $J_1$-$J_2$ model onto the canonical $J$-$h$ field model---with $n=1,2,3$ being the Ising, XY, and Heisenberg classical spin models, respectively; the theorem resolved a longstanding challenge for $n=3$ published in 1990. Its proof was done with an AI-synthesized recursive Householder moving frame and understood via a
human-recognized hidden reciprocity. Its generalizations to arbitrary interaction functions and arbitrary inhomogeneity were recognized by the AI and the human, respectively.  
An analogous theorem holds when the continuous $O(n)$ spins are replaced by the $q$-state Potts spins, implying a closed-form exact solution of the $J_1$-$J_2$ standard Potts open chain for every $q\ge2$ and every $L\ge1$. The emergence of these theorems from a human–AI co-development framework suggests that sustained AI involvement throughout a systematic research program, coupled with the human’s growing understanding of the implications of AI-generated results, may incubate autonomous scientific breakthroughs and make aspects of the discovery process experimentally testable.
\end{abstract}

\maketitle


\section{Introduction\label{sec:intro}}

Two questions about scientific breakthroughs driven by human--AI collaboration motivate this work. The first concerns scientific understanding. As general-purpose reasoning AI moves from answering well-specified questions toward generating unexpected hypotheses and solutions~\cite{Yin_Potts_J1-J2_1D,Cheng_NatMater_26_AI_Materials,Okabe_NatMater_26_SCIGEN,Schoener_JMMM_26_Materials_Discovery,Gottweis_25_AI_Coscientist,Penades_25_AI_Bacteria,9d_OpenAIUnitDistance_26,9d_Sawin_26,9d_Alon_26}, the practical bottleneck may invert: instead of asking only whether an AI understands a human scientist's problem, the human scientist must also recognize and understand what the AI has produced~\cite{Krenn_NRP_22,Ceriotti_PRXI_26}. An AI response can mix hallucination, algebraic error, rediscovery, irrelevant information, useful connection, and genuinely new insight in the same channel. The scientist must distinguish among them and determine which implications survive derivation, falsification, comparison with prior work, and physical interpretation. Once the realistic results are identified and validated, can the human scientist fully understand their implications, especially those beyond the active hypothesis space (AHS)~\cite{Klahr_88_AHS}, effectively leading to next breakthroughs? 

The second question is methodological: how can human--AI research be organized systematically to increase the chance of scientific breakthroughs, such as solving longstanding puzzles or revealing unknown unknowns far beyond the human's AHS? This question complements recent efforts toward AI co-scientists~\cite{
Gottweis_25_AI_Coscientist}. Recently, we developed an AI-aided research program, in which four projects~\cite{Yin_Potts_J1-J2_1D,Yin_site,Yin_Potts_UNPC,yin_7_CMM_arXiv} were deliberately organized so that the AI's contribution categories spanning 1d--7d in terms of a nine-dan AI contribution framework~\cite{Yin_Potts_UNPC} could be examined as relatively clean cases. Once the designated reasoning was finished, AI did not remain embedded through every subsequent stage of the scientific project. That design was useful for attribution, but it effectively froze the AI role after the assigned contribution. Built upon these experiments, the present work was designed to take an opposite approach: AI remains involved from the beginning to the end of a research trajectory---through derivation, comparison with literature, generalization, criticism, falsification, connection, and manuscript development. This multidimensional setting allows a different possibility to be tested: can the AI's role itself escalate to 8d/9d as the shared scientific landscape accumulates?

Here, we show that addressing the two questions (Sec.~\ref{sec:organizing}) has led to remarkable new AI-led insights, far beyond the human's AHS, into longstanding problems in statistical mechanics and magnetic materials. The two kinds of physics models at the center of this history are themselves foundational. The first kind is the so-called $n$-vector spin model with nearest-neighbor (NN) exchange interaction $J$~\cite{Stanley_PR_69_vector} ($n=1,2,3$ are the classical Ising, XY, and Heisenberg models, respectively~\cite{Mattis_book_1981,Mattis_book_1985,Mattis_book_08_SMMS}), which is the canonical setting for studying collective behavior arising in zero field and magnetic response to an external magnetic field $h$, where exchange-driven collective correlations compete with field-driven polarization and determine magnetization and field-dependent thermodynamics. For spins forming a one-dimensional (1D) lattice, it is described by the $J$-$h$ Hamiltonian
\begin{equation}
\mathcal H_{J-h}
=-J\sum_i{\bm S}_i\cdot{\bm S}_{i+1}
-h\sum_i{\bm S}_i^z,
\label{eq:J-h}
\end{equation}
where ${\bm S}_i$ is a unit vector in the $n$-dimensional spin space at lattice site $i$. The \(q\)-state planar Potts model, also known as the clock model, is a discrete \(n=2\) model with \(\mathbb{Z}_q\) rotational symmetry, whose simplest case, \(q=2\), is the Ising model~\cite{Potts_1952}. In contrast, the standard \(q\)-state Potts model generalizes the Ising model (\(q=2\)) by allowing the full permutation symmetry \(\mathbb{S}_q\)~\cite{Potts_1952,Potts_RMP_82}. Hereafter, the unqualified term Potts model refers to the standard \(q\)-state Potts model. The second kind is the so-called $J_1$-$J_2$ spin model, which is a minimal archetype of frustrated magnetism: interactions at two competing ranges---NN and next-nearest-neighbor (NNN)---can favor incompatible local arrangements and generate nontrivial correlations, modulated or helical ground states, dimerization, and degeneracy across discrete and continuous spin settings~\cite{Kivelson_24_book_statistical,Ramirez_review_frustrated_magnets_94,2011_book_frustration}. For a 1D chain, the $J_1$-$J_2$ model is described by
\begin{equation}
\mathcal H_{J_1-J_2}
=-J_1\sum_i{\bm S}_i\cdot{\bm S}_{i+1}
-J_2\sum_i{\bm S}_i\cdot{\bm S}_{i+2}.
\label{eq:J1-J2}
\end{equation}
The corresponding transfer-matrix or transfer-integral problems for calculating the partition function can be formulated exactly~\cite{Kramers_Wannier_PR_41_transferMatrix} and evaluated numerically~\cite{Glumac_JPA_93_Potts_LR,Blume_PRB_75_1D_Classical_Heisenberg_field,Harada_ZPB_88_1D_J1-J2_Classical_Heisenberg}. Of particular interest for these textbook 1D systems is whether exact analytical results can be achieved to provide classroom-level deep understanding. Table~\ref{table:problems} ($J$-$h$ and $J_1$-$J_2$ lines) summarizes a representative status of exact closed-form solutions for these chain models, showing that most of them are longstanding unsolved problems. \begin{table}[b]
\caption{Representative status of exact closed-form solutions and $J_1$-$J_2$ $\leftrightarrow$ $J$-$h$ mapping for the spin chains. T-map and H-map stand for the thermodynamic-limit and Hamiltonian-level mappings,
respectively. X denotes no exact closed-form solution known to us. I--IV denote four types of discovery demonstrated in this work (see Table~\ref{tab:types}).}
\begin{tabular}{c|ccccc}
\hline\hline
Model&Ising&Potts&clock/XY&Heisenberg&$O(n)$\\
&($n=1$)&&($n=2$)&($n=3$)&\\
\hline
$J$ 
& 1925~\cite{Ising1925} 
& 1952~\cite{Potts_1952} 
& 1967~\cite{Joyce_PRL_67_1D_XY}
& 1964~\cite{Fisher_64_1D_Heisenberg_classical} 
&1969~\cite{Stanley_PR_69_vector} \\
$J$-$h$ 
& 1941~\cite{Kramers_Wannier_PR_41_transferMatrix} 
& 1994~\cite{Glumac_JPA_94_Potts_J1-h} 
& X\ignore{~\cite{Loveluck_JPC_79_1D_J-h_XY}} 
& X\ignore{~\cite{Blume_PRB_75_1D_Classical_Heisenberg_field}} & X\\
$J_1$-$J_2$ 
& 1969~\cite{Dobson_JMathP_69_Many-Neighbored-Ising-Chain} 
& 2025~\cite{Yin_Potts_J1-J2_1D} 
& X\ignore{~\cite{Harada_JPSJ_84_1D_J1-J2_XY}}
&
X\ignore{~\cite{Harada_ZPB_88_1D_J1-J2_Classical_Heisenberg,Harada_JPC_90_1D_J1-J2_Classical_Heisenberg_XY}} 
& X \\ \hline
T-map & 1969~\cite{Dobson_JMathP_69_Many-Neighbored-Ising-Chain} & 2025~\cite{Yin_Potts_J1-J2_1D} & \footnote{The clock result follows immediately by discrete restriction of the planar mapping (see Appendix~\ref{appendix:special_n2}).}1984~\cite{Harada_JPSJ_84_1D_J1-J2_XY} & 1988~\cite{Harada_JPhysiqueColl_88} & III $\to$ II\\
H-map & 1969~\cite{Dobson_JMathP_69_Many-Neighbored-Ising-Chain} & IV & \footnote{H-map for $n=2$ was declared in Ref.~\cite{Harada_JPC_90_1D_J1-J2_Classical_Heisenberg_XY} but its explicit derivation was not given (see Sec.~\ref{sec:reflection}).\label{footnote:XY}}1990\cite{Harada_JPC_90_1D_J1-J2_Classical_Heisenberg_XY}? & I $\to$ IV & II $\to$ IV\\
\hline\hline
\end{tabular}
\label{table:problems}
\end{table}
The possibility of extracting significant new physical insights from these models was revived by the 2025 AI-bootstrapped exact solution of the $J_1$-$J_2$ Potts chain for every $q\ge 2$ in the thermodynamic limit $L\to\infty$~\cite{Yin_Potts_J1-J2_1D}, where $L$ denotes the number of spins, i.e., the chain size.

The two kinds of models are normally treated as separate physical stories: One is about geometric frustration, the other about symmetry-breaking fields. In the transfer-matrix formulation, the range-two $J_1$-$J_2$ chain appears much more complicated than the range-one $J$-$h$ chain: 
For a $q$-state Potts or clock realization, the order of the transfer matrix of the $J_1$-$J_2$ chain is $q^2$, compared with $q$ for the $J$-$h$ chain~\cite{Yin_Potts_J1-J2_1D,Glumac_JPA_93_Potts_LR,Glumac_JPA_94_Potts_J1-h}; for continuous $O(n)$ spins, the transfer problems for the $J_1$-$J_2$ and $J$-$h$ chains are represented in pair-spin and single-spin configuration spaces, respectively~\cite{Blume_PRB_75_1D_Classical_Heisenberg_field,Harada_JPSJ_84_1D_J1-J2_XY,Harada_ZPB_88_1D_J1-J2_Classical_Heisenberg,Harada_JPhysiqueColl_88,Harada_JPC_90_1D_J1-J2_Classical_Heisenberg_XY}. However, an exact thermodynamic-limit free-energy mapping between the $J_1$-$J_2$ and $J$-$h$ chains has been established for the $n\le3$ and Potts cases,
\ignore{
by matching the partition functions (as functionals of the Hamiltonians),
\begin{equation}
\mathbb{Z}\left[\mathcal{H}_{J_1-J_2}\right]= c\;\mathbb{Z}[\mathcal{H}_{J-h}]\;\;\; \mathrm{for}\; L\to\infty,
\label{eq:T-map}
\end{equation}
where $J=J_2$, $h=J_1$ and $c$ is a constant,
} 
as shown in Table~\ref{table:problems} (T-map line). The implications of these results are remarkable: the exact 
bridge between the two canonical models not only reduces a seemingly much more complicated problem to a simpler one, but also connects two basic organizing mechanisms of statistical physics, i.e., collective behavior arising spontaneously 
from competing interactions and collective behavior induced or controlled by an external field.   

Moreover, a Hamiltonian-level mapping was found for \(n=1\) (Ising)~\cite{Dobson_JMathP_69_Many-Neighbored-Ising-Chain} and arguably \(n=2\)~\cite{Harada_JPC_90_1D_J1-J2_Classical_Heisenberg_XY} open chains, as shown in Table~\ref{table:problems} (H-map line). Unlike the thermodynamic mapping or an approximate surrogate, the Hamiltonian-level mapping establishes a temperature-independent configuration-level bijection that preserves both the Hamiltonian and equilibrium measure for every finite open-chain length \(L\), thereby carrying over the entire finite-size equilibrium energy landscape and all correspondingly transformed observables exactly; the thermodynamic free-energy mapping follows merely as its \(L\to\infty\) corollary.
For \(n=3\), Harada and Mikeska found an exact partition-function mapping for arbitrary finite \(L\) through an explicit comparison of the transfer-integral calculations;
however, they explicitly argued that this duality could not be derived at the Hamiltonian level because rotations in three dimensions do not commute~\cite{Harada_JPC_90_1D_J1-J2_Classical_Heisenberg_XY}.
Can we overcome this 1990 stated obstruction and derive a Hamiltonian-level mapping for \(n=3\) and, more generally, every \(n\ge1\)? If successful, what implications might this cleanest mapping carry?

Table~\ref{table:problems} (H-map line) also reveals a knowledge gap: a Hamiltonian-level mapping has not been published for the standard $q$-state Potts chain yet, to the best of our knowledge. 
Can we close this knowledge gap? The fact that the $J$-$h$ Potts open chain has been solved in closed form~\cite{Glumac_JPA_94_Potts_J1-h,ChangShrock2009} makes the question particularly interesting: if found, such a microscopic mapping would yield a closed-form exact solution to the $J_1$-$J_2$ Potts open chain for every $q\ge2$ and every $L\ge1$, a clear leap from the 2025 $L\to\infty$ result~\cite{Yin_Potts_J1-J2_1D}.

\begin{figure}[t]
    \begin{center}
\includegraphics[width=0.8\columnwidth,clip=true,angle=0]{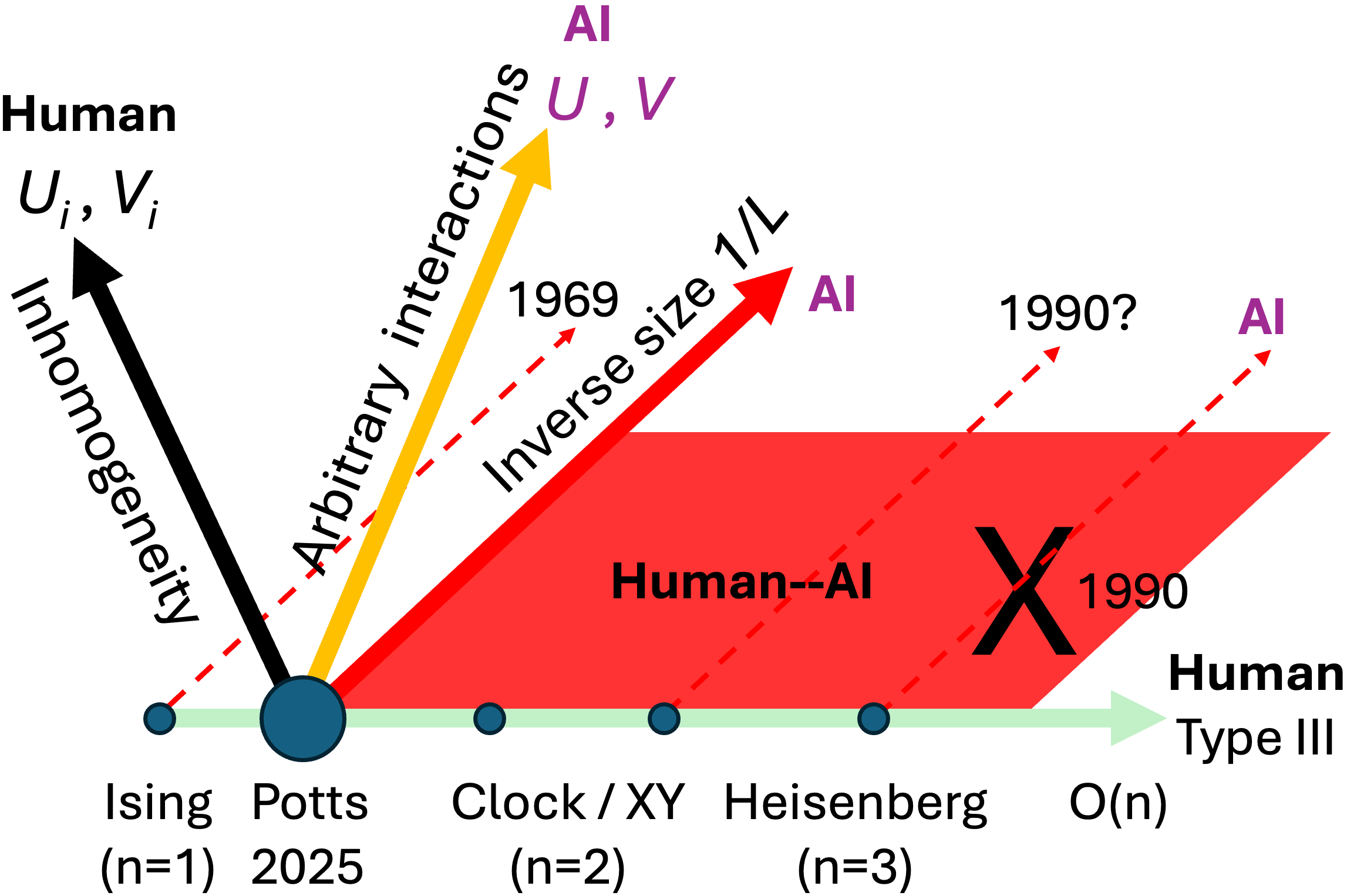}
    \end{center}
\caption{Illustration of a human--AI research developed from a single point (big blue dot~\cite{Yin_Potts_J1-J2_1D}) to a four-axis AHS: The green axis stands for model choice, red for inverse system size $1/L$, and orange for homogeneous arbitrary interactions $U(x),V(x)$, and black for inhomogeneous $U_i(x), V_i(x)$. The four AHS axes are organizational rather than four scalar coordinates (see text).
The main driver of  different subspaces is indicated as Human or AI. Human--AI means that they drove the $n-L$ subspace together. The initial AHS is represented by the green arrow and marked as Type III. The 1969 and 1990 red dashed lines stand for the Hamiltonian-level mapping for the Ising model~\cite{Dobson_JMathP_69_Many-Neighbored-Ising-Chain} and that declared for the XY model~\cite{Harada_JPC_90_1D_J1-J2_Classical_Heisenberg_XY} (see Table~\ref{table:problems}, footnote~\ref{footnote:XY}), respectively. X denotes the 1990 obstruction~\cite{Harada_JPC_90_1D_J1-J2_Classical_Heisenberg_XY}. 
}
\label{fig:AHS}
\end{figure}

In this article, we present a human--AI collaboration organized into a co-development framework, in which the human-designed initial AHS for the AI's progressive development consisted of the thermodynamic free-energy mapping problems between the $J_1$-$J_2$ and $J$-$h$ chains, with challenge progression through one axis (Fig.~\ref{fig:AHS}, green axis): the $q$-state Potts, $q$-state clock, $n=2$, $n=3$, and $O(n)$ models. 
The AI participation then continued instead of terminating after a successful reasoning. 
The result is a theorem for a four-axis AHS (Fig.~\ref{fig:AHS}) far beyond the initial one-axis AHS: an exact Hamiltonian-level mapping between the range-two and range-one open chains with arbitrary inhomogeneous interaction functions---not merely the homogeneous $J_1$-$J_2$ vs $J$-$h$ linear interactions---for every $n\ge1$ or every $q\ge2$, and every $L\ge1$. 
The four AHS axes should not be interpreted as four ordinary scalar coordinates. Within the \(O(n)\) family, model choice forms a 1D discrete axis indexed by \(n\), whereas the Potts and clock families introduce separate \(q\)-parameterized branches. The inverse-size axis is 1D, while the arbitrary-interaction axis is already an infinite-dimensional function space; allowing independent variation promotes it to an \(O(L)\)-fold product of such spaces. Thus the four-axis AHS is organizationally compact but mathematically vast, making it especially suitable for further AI-assisted exploration. This work exhibited two 9d AI-led episodes under the definitions of Ref.~\cite{Yin_Potts_UNPC}. 

Yet, the final extension to arbitrary inhomogeneity was discovered by the human without modifying the proof. This not only highlights the critical role of human recognition and understanding of the implications of AI-generated results, but also brings a much broader class of experimentally relevant inhomogeneous geometries within the scope of the theorem (Sec.~\ref{sec:experiment}). 

The rest of this article is organized as follows: Sec.~\ref{sec:theorem} presents the Hamiltonian-level theorem for the $O(n)$-vector-spin open chains, its proof done with an AI-synthesized recursive Householder moving frame and understood via a human-recognized hidden reciprocity obtained from historical context analysis, and its implications including  the linear corollary, a physical nonlinear example (single-ion anisotropy $\longleftrightarrow$ biquadratic exchange), and broader experimental and theoretical relevance. Sec.~\ref{sec:potts} presents an analogous theorem for the $q$-state Potts open chains, its proof, and an exact closed-form solution to the $J_1$-$J_2$ standard Potts chain for every $q\ge 2$ and $L\ge 1$. Sec.~\ref{sec:discussion} addresses how this work was prepared and organized into a human--AI co-development framework (including categorizations of four discovery types, nine-dan AI contribution, and being vs non-being prompting strategies), 
the research trajectory based on the timestamped human--AI conversation transcript~\cite{data:9d}, 
comparison with other AI discovery modes, and limitations and testable hypotheses for future AI research.

\section{Theorem,  Proof, and Implication\label{sec:theorem}}

Consider an open chain of $L$ classical unit vectors
\begin{equation}
 {\bm S}_i\in S^{n-1},\qquad n\ge 1,
\end{equation}
with the range-two Hamiltonian
\begin{equation}
\mathcal H^{\rm OBC}_{\rm range\,2}(L)
=\sum_{i=1}^{L-1}
U_i(\bm S_i\cdot\bm S_{i+1})
+\sum_{i=1}^{L-2}V_i(\bm S_i\cdot\bm S_{i+2}),
\label{eq:range2}
\end{equation}
where $U_i(x)$ and $V_i(x)$ are arbitrary real functions on $x\in [-1,1]$ and may vary arbitrarily with the site index $i$, allowing the interactions to be independently inhomogeneous along the chain. OBC means the open boundary condition. Empty sums are understood to vanish.

The corresponding range-one open chain with another set of $L-1$ classical unit vectors
${\boldsymbol\sigma}_i\in S^{n-1}$ is
\begin{equation}
\mathcal H^{\rm OBC}_{\rm range\,1}(L-1)
=\sum_{i=1}^{L-1}U_i(\hat{\bm z}\cdot\boldsymbol\sigma_i)+\sum_{i=1}^{L-2}
V_i(\boldsymbol\sigma_i\cdot\boldsymbol\sigma_{i+1}),
\label{eq:range1}
\end{equation}
where \(\hat{\bm z}\in S^{n-1}\) is an arbitrary fixed reference unit vector.
The NN interaction $U_i$ of the range-two model becomes an axial single-spin potential, whereas the NNN interaction $V_i$ becomes the NN interaction. 

\begin{theorem} 
For every positive integer $n$ and every chain size $L$, there exists a bijection,
\begin{equation}
\{{\bm S}_1,{\bm S}_2,{\bm S}_3,\dots,{\bm S}_L\}\longleftrightarrow \{{\bm S}_1,\boldsymbol\sigma_1,\boldsymbol\sigma_2,\dots,\boldsymbol\sigma_{L-1}\},   
\label{eq:bijection}
\end{equation}
satisfying
\begin{eqnarray}
   \bm S_i\cdot\bm S_{i+1}
 &=&\hat{\bm z}\cdot\boldsymbol\sigma_i,\label{eq:geom_nn_identity}\\
 \bm S_i\cdot\bm S_{i+2}
&=&\boldsymbol\sigma_i\cdot\boldsymbol\sigma_{i+1}.
\label{eq:geom_nnn_identity}  
\end{eqnarray}
Then, for arbitrary inhomogeneous functions $U_i(x),V_i(x)$ on $x\in[-1,1]$, 
Eqs.~\eqref{eq:range2} and \eqref{eq:range1} can be mapped exactly onto each other under the change of variables, 
\begin{equation}
 \mathcal H_{\rm range\, 2}^{\rm OBC}(L)
 ={\mathcal H}_{\rm range\,1}^{\rm OBC}(L-1).
 \label{eq:geom_H_identity}
\end{equation}
The Hamiltonian mapping is temperature independent, and at every finite temperature for which the corresponding Boltzmann integrals exist, the partition functions obey the exact identity
\begin{equation}
 Z_{\mathrm{range}\,2}^{\rm OBC}(L)
 =Z_{\mathrm{range}\,2}^{\rm OBC}(1)\cdot Z_{\mathrm{range}\,1}^{\rm OBC}(L-1),
 \label{eq:OBC_partition_identity}
\end{equation}
with $Z_{\mathrm{range}\,1}^{\rm OBC}(0)=1$ and the constant $Z_{\mathrm{range}\,2}^{\rm OBC}(1)$ being the partition function of a single spin ($\bm S_1$) in zero field.
\label{thm:UiVi}
\end{theorem}


\subsection{Proof: A recursive Householder moving frame\label{sec:householder}}
Householder reflections are standard symmetric orthogonal transformations in numerical linear algebra~\cite{Householder1958,GolubVanLoan2013}, while moving-frame constructions have a long geometric history~\cite{FelsOlver1999}.  Hyperplane reflections also appear as symmetry operations in Monte Carlo algorithms for classical $O(n)$ spins~\cite{Wolff1989,KentDobiasSethna2018}.  We use the phrase \emph{recursive Householder moving frame} descriptively for the construction below: the frame is propagated site by site through a product of local Householder reflections.  To our knowledge, the recursive use of Householder reflections as an exact change of spin variables that reduces a 1D range-two Hamiltonian to a range-one Hamiltonian has not previously been formulated.

First, fix once and for all a reference unit vector \(\hat{\bm e}
\in S^{n-1}\), which we call the starting \(z\)-axis. 
Then, for every $\bm x\in S^{n-1}$ define a Householder reflector as the $n\times n$ symmetric orthogonal matrix $\mathsf H(\bm x)\in O(n)$ that sends $\hat{\bm e}$ to $\bm x$,
\begin{equation}
 \mathsf H(\bm x)=
 \begin{cases}
 \displaystyle
 \mathsf I-
 \frac{(\hat{\bm e}-\bm x)(\hat{\bm e}-\bm x)^{\mathsf T}}
 {1-\hat{\bm e}\cdot\bm x},
 & \bm x\neq \hat{\bm e},\\[3mm]
 \mathsf I,&\bm x=\hat{\bm e}.
 \end{cases}
 \label{eq:householder_section}
\end{equation}
Since $\|\hat{\bm e}-\bm x\|^2=2(1-\hat{\bm e}\cdot\bm x)$, Eq.~\eqref{eq:householder_section} is the standard Householder form $\mathsf I-2\bm v\bm v^{\mathsf T}/(\bm v^{\mathsf T}\bm v)$ with $\bm v=\hat{\bm e}-\bm x$.  It obeys
\begin{equation}
 \mathsf H(\bm x)^{\mathsf T}=\mathsf H(\bm x),\;\;\;
 \mathsf H(\bm x)^{\mathsf T}\mathsf H(\bm x)=\mathsf I,\;\;\;
 \mathsf H(\bm x)\hat{\bm e}=\bm x.
 \label{eq:householder_properties}
\end{equation}
The special definition $\mathsf H(\hat{\bm e})=\mathsf I$ is sufficient and it is the conventional setting for algorithmic efficiency; no globally smooth choice of frame is required for the change of variables or the partition-function identity. Other choices satisfying Eq.~(\ref{eq:householder_properties}) at $\bm x=\hat{\bm e}$, e.g., a reflection across a hyperplane containing $\hat{\bm e}$, may allow a continuous directional approach with $\mathsf{det\; H}(\hat{\bm e})=-1$ (Appendix~\ref{appendix:special_n2}). 

The AI used Eq.~(\ref{eq:householder_properties}) together with a recursive moving frame to construct the bijection, Eqs.~(\ref{eq:bijection})--(\ref{eq:geom_nnn_identity}). Choose
\begin{equation}
 \mathsf G_1=\mathsf H(\bm S_1),
 \qquad \mathsf G_1\hat{\bm e}=\bm S_1,
\end{equation}
and define recursively, for $i=1,\ldots,L-1$,
\begin{equation}
 \boldsymbol\sigma_i=\mathsf G_i^{\mathsf T}\bm S_{i+1},
 \qquad
 \mathsf G_{i+1}=\mathsf G_i\mathsf H(\boldsymbol\sigma_i).
 \label{eq:recursive_householder}
\end{equation}
Explicitly, $$\mathsf G_{i}=\mathsf H(\bm S_1)\mathsf H(\boldsymbol\sigma_1)H(\boldsymbol\sigma_2)\cdots\mathsf H(\boldsymbol\sigma_{i-1}).$$ 
Because $\mathsf G_i^{\mathsf T}\mathsf G_i=\mathsf G_i\mathsf G_i^{\mathsf T} =\mathsf I$, $\mathsf G_i$ is orthogonal.  Thus, $\boldsymbol\sigma_i^{\mathsf T}\boldsymbol\sigma_i=\bm S_{i+1}^{\mathsf T}\bm S_{i+1}=1$; each $\boldsymbol\sigma_i$ is again a unit $n$-vector.  Moreover,
\begin{equation}
 \mathsf G_{i+1}\hat{\bm e}
 =\mathsf G_i\mathsf H(\boldsymbol\sigma_i)\hat{\bm e}
 =\mathsf G_i\boldsymbol\sigma_i
 =\bm S_{i+1},
\end{equation}
so induction gives
\begin{equation}
 \bm S_i=\mathsf G_i\hat{\bm e}
 \qquad (i=1,\ldots,L).
 \label{eq:frame_tracks_spin}
\end{equation}
Geometrically, the recursive moving frame rotates the coordinate system so that the preceding spin becomes the local $z$-axis, and then describes the next spin in that frame. Specifically, at step \(i\), \(\bm S_i\) defines the local $z$-axis and
\(\boldsymbol\sigma_i=\mathsf G_i^{\mathsf T}\bm S_{i+1}\)
describes \(\bm S_{i+1}\) in this local frame. The Householder update then constructs the next frame so that \(\bm S_{i+1}\) becomes its local $z$-axis.

The two identities, Eqs.~\eqref{eq:geom_nn_identity} and \eqref{eq:geom_nnn_identity}, responsible for the Hamiltonian-level mapping are now immediate.  First,
\begin{align}
 \bm S_i\cdot\bm S_{i+1}
 &=\hat{\bm e}^{\mathsf T}\mathsf G_i^{\mathsf T}\,\mathsf G_i
  \mathsf H(\boldsymbol\sigma_i)\hat{\bm e}
  =\hat{\bm e}^{\mathsf T}
  \mathsf H(\boldsymbol\sigma_i)\hat{\bm e}
 =\hat{\bm e}\cdot\boldsymbol\sigma_i.
 \label{eq:geom_nn_identity2}  
\end{align}
Second, using
$\bm S_{i+2}=\mathsf G_i\mathsf H(\boldsymbol\sigma_i)
\mathsf H(\boldsymbol\sigma_{i+1})\hat{\bm e}$,
\begin{align}
 \bm S_i\cdot\bm S_{i+2}
 &=\hat{\bm e}^{\mathsf T}\mathsf G_i^{\mathsf T}\,\mathsf G_i
 \mathsf H(\boldsymbol\sigma_i)
 \mathsf H(\boldsymbol\sigma_{i+1})\hat{\bm e}
 \nonumber\\
 &=\bigl[\mathsf H(\boldsymbol\sigma_i)\hat{\bm e}\bigr]^{\mathsf T}
   \bigl[\mathsf H(\boldsymbol\sigma_{i+1})\hat{\bm e}\bigr]
 \nonumber\\
 &=\boldsymbol\sigma_i\cdot\boldsymbol\sigma_{i+1}.
 \label{eq:geom_nnn_identity2}  
\end{align}
Importantly, the second identity uses the symmetry
$\mathsf H^{\mathsf T}=\mathsf H$, not commutativity between successive transformations.  Hence the Householder construction does not require three-dimensional rotations, or their higher-dimensional counterparts, to commute.

The transformation is invertible for an open chain: given $\bm S_1$ and $\{\boldsymbol\sigma_i\}_{i=1}^{L-1}$, Eqs.~\eqref{eq:recursive_householder} and \eqref{eq:frame_tracks_spin} reconstruct all original spins uniquely, thereby giving the bijection, Eqs.~(\ref{eq:bijection})--(\ref{eq:geom_nnn_identity}).

Renaming $\hat{\bm e}=\hat{\bm z}$ and substituting Eqs.~\eqref{eq:geom_nn_identity} and \eqref{eq:geom_nnn_identity} into Hamiltonian~(\ref{eq:range2}) derives Hamiltonian~(\ref{eq:range1}), because the proof acts on the geometric arguments of the interactions, not on the interactions themselves. This gives the Hamiltonian identity, Eq.~(\ref{eq:geom_H_identity}).
Thus a zero-field range-two $O(n)$ chain of $L$ spins is mapped exactly onto a simpler range-one $O(n)$ chain of $L-1$ spins with an axial single-spin potential. 
 
Importantly, the same change of variables also preserves the open-chain measure microscopically.  Let $d\mu_n$ denote the normalized $O(n)$-invariant measure on $S^{n-1}$.  For fixed preceding variables, $\bm S_{i+1}=\mathsf G_i\boldsymbol\sigma_i$ is only an orthogonal transformation, and therefore
\begin{equation}
 d\mu_n(\bm S_{i+1})=d\mu_n(\boldsymbol\sigma_i).
 \label{eq:measure}
\end{equation}
Consequently, performing the change of variables successively along the chain yields
\begin{equation}
 \prod_{i=1}^{L}d\mu_n(\bm S_i)
 =d\mu_n(\bm S_1)
 \prod_{i=1}^{L-1}d\mu_n(\boldsymbol\sigma_i).
 \label{eq:measure_factorization}
\end{equation}
The partition function
\begin{eqnarray}
Z_{\rm range\,2}^{\rm OBC}(L)
&=&
\int \prod_{i=1}^{L}d\mu_n(\bm S_i)
e^{-\beta \mathcal H_{\rm range\,2}^{\rm OBC}(L)
} \nonumber\\
&=&\int d\mu_n(\bm S_1)\left[\int
 \prod_{i=1}^{L-1} d\mu_n(\boldsymbol\sigma_i)
e^{-\beta \mathcal H_{\rm range\,1}^{\rm OBC}(L-1)}\right]\nonumber\\
&=&Z_{\mathrm{range}\,2}^{\rm OBC}(1)\cdot Z_{\rm range\,1}^{\rm OBC}(L-1),
\label{eq:J1J2-n-OBC-Z}
\end{eqnarray}
where $\beta=1/(k_\mathrm{B}T)$ with $k_\mathrm{B}$ being the Boltzmann constant and $T$ the absolute temperature.
Since $\mathcal H_{\rm range\,1}^{\rm OBC}(L-1)$ is independent of $\bm S_1$, the integral with respect to $\mu_n(\bm S_1)$ yields a factor that is nothing but $Z_{\mathrm{range}\,2}^{\rm OBC}(1)$, which is unity if normalized $d\mu_n(\bm S_1)$ is used.
If instead the unnormalized surface measure $d\Omega$ is used, $Z_{\mathrm{range}\,2}^{\rm OBC}(1)$ is the global-orientation factor
$\Omega_{n-1}=\frac{2\pi^{n/2}}{\Gamma(n/2)}$. Note 
$\Omega_{n-1}=2,2\pi, 4\pi$ for $n=1,2,3$, respectively. 
\hfill$\blacksquare$ 

For periodic boundaries, the local Hamiltonian identities remain valid, but the
transformed spins are no longer independent.  Closing the original chain by
$\bm S_{L+1}=\bm S_1$ imposes the global holonomy constraint
\begin{equation}
 \mathsf H(\boldsymbol\sigma_1)
 \mathsf H(\boldsymbol\sigma_2)\cdots
 \mathsf H(\boldsymbol\sigma_L)\hat{\bm e}
 =\hat{\bm e}.
 \label{eq:periodic_holonomy}
\end{equation}
Thus, the periodic problem is mapped exactly only onto the constrained subset of
dual-spin configurations satisfying Eq.~\eqref{eq:periodic_holonomy}, rather than
onto an unconstrained periodic $J$-$h$ chain.  For $n=1$, the constraint reduces to
$\prod_i \sigma_i=1$~\cite{Mueller_NPB_17_PBC_Ising}, while for $n=2$ it becomes
the bond-angle closure condition
$\sum_i(\Theta_{i+1}-\Theta_i)=0\pmod{2\pi}$,
equivalently an integer winding number around the ring~\cite{Cosco_JMP_21_Periodic_XY}.
Accordingly, the unconstrained finite-size periodic partition functions 
need not be equal, even though the local Hamiltonian mapping is exact.  This distinction is immaterial for
the bulk thermodynamics: for bounded finite-range interactions, changing between open and periodic boundaries modifies the total free energy only by a boundary contribution, so the two boundary conditions have the same thermodynamic-limit free-energy density~\cite{Dobson_JMathP_69_Many-Neighbored-Ising-Chain}.  Therefore Eq.~\eqref{eq:OBC_partition_identity} immediately implies the thermodynamic-limit (\(L\to\infty\)) free-energy mapping. 

\subsection{Low-dimensional cases and the hidden reflection structure\label{sec:reflection}}

The low-dimensional special cases sharpen both the historical context---within the homogeneous linear specialization to the $J_1$-$J_2$ and $J$-$h$ chains---and the geometric content of the present construction. Explicit reductions of the recursive Householder moving frame for $n=1,2,3$ are given in Appendix~\ref{appendix:n123}. Here we focus on understanding the AI-generated results: In hindsight, can the historical context help us identify an insightful feature of the recursive Householder moving frame rather than accepting it merely as a technical construction?
This hindsight is highly important because it provides an encouraging case that the AI-generated breakthrough is not only understandable but also likely retrospectively achievable by the human.

\subsubsection{Historical context}

For $n=1$, $S^0=\{\pm1\}$. With $\hat{\bm e}=1$, the $1\times1$ Householder matrices reduce to the scalars: $\mathsf H(\sigma_i)=\sigma_i$ and $\mathsf G_i=S_i$.  Equation~\eqref{eq:recursive_householder} then gives 
\begin{equation}
 \sigma_i=S_iS_{i+1},  \label{eq:icu_n1} 
\end{equation} 
while Eq.~\eqref{eq:geom_nnn_identity2} reduces to Dobson's Ising bond identity $\sigma_i\sigma_{i+1}=S_iS_{i+2}$~\cite{Dobson_JMathP_69_Many-Neighbored-Ising-Chain}.  Thus, the apparently special algebraic property $S_i^2=1$ is the $O(1)$ endpoint of the general geometric construction.  Its reflection content is, however, almost invisible: in one dimension an orthogonal reflection is merely multiplication by $-1$, so the geometry was naturally perceived only as a sign reversal.  

For $n=2$, Harada introduced relative planar angles for a periodic chain and obtained the corresponding transfer-integral formulation~\cite{Harada_JPSJ_84_1D_J1-J2_XY}; because of the periodic closure constraint Eq.~(\ref{eq:periodic_holonomy}), this does not constitute an exact finite-size Hamiltonian mapping.  Harada and Mikeska later treated open chains and established the exact finite-size partition-function equivalence~\cite{Harada_JPC_90_1D_J1-J2_Classical_Heisenberg_XY}.
\ignore{,  
\begin{equation*}
\mathbb{Z}\left[\mathcal{H}_{J_1-J_2}^{\rm OBC}
 (L)\right]=2\pi\,\mathbb{Z}[\mathcal{H}_{J-h}^{\rm OBC}
 (L-1)]
\label{eq:T-map-OBC}
\end{equation*}
with $J=J_2$ and $h=J_1$.}
Although they subsequently referred to the planar duality as derivable at the Hamiltonian level,
\ignore{
\begin{equation*}
 \mathcal H_{J_1-J_2}^{\rm OBC}
 (L)
 ={\mathcal H}_{J-h}^{\rm OBC}
 (L-1),
 \label{eq:H-map}
\end{equation*}}
neither of these papers displays an explicit derivation of the  Hamiltonian-level mapping.  In such a mapping, the operation $\phi\mapsto-\phi$ was naturally read in angular language as reversing the sense of rotation, although in Cartesian coordinates it is exactly reflection about the $x$ axis.  Thus, in both the $n=1$ and $n=2$ cases, the reflection structure is present but geometrically inconspicuous.

For $n=3$, Harada and Mikeska established the same duality for finite open chains at the transfer-integral level~\cite{Harada_JPC_90_1D_J1-J2_Classical_Heisenberg_XY} by augmenting the $n=2$ bond-angle picture with signed dihedral angles and accumulated azimuths as follows~\cite{Harada_ZPB_88_1D_J1-J2_Classical_Heisenberg,Harada_JPhysiqueColl_88}: Define the bond angle $\vartheta_i$ by 
\begin{equation}
\cos\vartheta_i=\bm S_i\cdot\bm S_{i+1},\label{eq:XY_sigmaz}
\end{equation}
and let $\varphi_i$ be the signed dihedral angle between the planes $(\bm S_{i},\bm S_{i+1})$ and $(\bm S_{i+1},\bm S_{i+2})$. Then to evaluate $\bm S_i\cdot\bm S_{i+2}$, take $\bm S_{i+1}$ as the local $z$ axis,
\begin{equation}
\bm S_i\cdot\bm S_{i+2}
=\cos\vartheta_i\cos\vartheta_{i+1}
+\sin\vartheta_i\sin\vartheta_{i+1}\cos\varphi_i.
\label{eq:torsion}
\end{equation}
Now introduce azimuth angles $\alpha_i$ recursively by \begin{equation}
    \alpha_{i+1}-\alpha_i=\varphi_i,
\end{equation} 
with arbitrary $\alpha_1$ (for example, $\alpha_1=0$)
and construct an auxiliary Heisenberg spin
\begin{equation}
\boldsymbol\sigma_i=(\sin\vartheta_i\cos\alpha_i,
\sin\vartheta_i\sin\alpha_i,\cos\vartheta_i).\label{eq:XY_sigma}
\end{equation}
Then
\begin{equation}
\sigma_i^z=\bm S_i\cdot\bm S_{i+1},\qquad
\boldsymbol\sigma_i\cdot\boldsymbol\sigma_{i+1}=\bm S_i\cdot\bm S_{i+2}.
\label{eq:Heisgeom}
\end{equation}
However, Harada and Mikeska explicitly argued that a Hamiltonian-level derivation was obstructed by the noncommutativity of three-dimensional rotations: 
\begin{quote}
``Owing to the non-commutativity of rotations in three dimensions this duality, in contrast to the corresponding property of the planar model, cannot be derived on the level of the Hamiltonian but only on the basis of an explicit comparison of the transfer integral calculations for the two models.''~\cite{Harada_JPC_90_1D_J1-J2_Classical_Heisenberg_XY}  
\end{quote}
The wording is revealing: (1) A Hamiltonian-level proof requires more than microscopic identities between interaction variables; it must establish an explicit change of variables and microscopic preservation of the measure, Eqs.~(\ref{eq:measure}) and (\ref{eq:measure_factorization}), with every required step justified mathematically rather than descriptively. The measure-preservation requirement was not established there; as presented, the bond-angle--dihedral-angle construction describes the transformation of local coordinate systems across sites without providing such a global change of variables together with its measure.
(2)  
Their stated obstruction was formulated entirely in terms of rotations, indicating that the $n=2$ planar construction was being viewed within a rotation-based geometric picture rather than through the reflection hidden in $\phi\mapsto-\phi$.  
(3)
The \(n=3\) bond-angle--dihedral-angle approach, as presented in Ref.~\cite{Harada_ZPB_88_1D_J1-J2_Classical_Heisenberg} and recently in Ref.~\cite{Dmitriev_PRB_19_delta_chain}, does not expose reflection as the organizing geometric operation, either. 
Three dimensions are special, too: two successive bond planes have a single signed dihedral angle, and one can encode that angle as the difference of two azimuths. For \(n>3\), there is no analogous single torsion angle that parametrizes the full relative orientation of the relevant transverse spaces. So 
the approach simply stops being obviously generalizable.
This interpretation is retrospective, but it explains why the known $n=1,2,3$ cases offered little guidance toward higher dimensions or toward an explicit unified Hamiltonian-level construction even for $n=2,3$.  

\subsubsection{Reciprocity as a hidden essence of the Householder-reflector construction\label{sec:reciprocity}}

To get insight into the AI's Householder construction, the human reexamined the bond-angle--dihedral-angle approach---now mathematically: The previous description of ``to evaluate $\bm S_i\cdot\bm S_{i+2}$, take $\bm S_{i+1}$ as the local $z$ axis'' can now be written as
\begin{eqnarray}
    \bm S_i\cdot\bm S_{i+2}&=&(\mathsf G_{i+1}^{\mathsf T} \bm S_{i})\cdot(\mathsf G_{i+1}^{\mathsf T} \bm S_{i+2})\label{eq:refresh}
\end{eqnarray}
using the $(i+1)$th frame $\mathsf G_{i+1}$. This equation holds if $\mathsf G_{i+1}\mathsf G_{i+1}^{\mathsf T}=\mathsf I$, i.e., $\mathsf G_{i+1}$ is orthogonal. Compared with the target identity $\bm S_i\cdot\bm S_{i+2}=\boldsymbol\sigma_i\cdot\boldsymbol\sigma_{i+1}$, Eq.~(\ref{eq:refresh}) suggests $\mathsf G_{i+1}^{\mathsf T} \bm S_{i}=\boldsymbol\sigma_i$ and $\mathsf G_{i+1}^{\mathsf T} \bm S_{i+2}=\boldsymbol\sigma_{i+1}$, which implies $\mathsf G_{i}^{\mathsf T} \bm S_{i+1}=\boldsymbol\sigma_{i}$ under a one-site translation. The human therefore identified a reciprocity encoded in 
\begin{eqnarray}
   \boxed{\mathsf G_{i+1}^{\mathsf T} \bm S_{i}=\mathsf G_{i}^{\mathsf T}\bm S_{i+1}}=\boldsymbol\sigma_i,
   \label{eq:reciprocity}
\end{eqnarray}
which means the local representation of $\bm S_{i}$ in the coordinate frame of $\bm S_{i+1}$ is identical to that of $\bm S_{i+1}$ in the frame of $\bm S_{i}$. 

By contrast, if the update from \(\mathsf G_i\) to \(\mathsf G_{i+1}\) were constructed purely by a proper rotation taking the new spin to the local \(z\)-axis, then geometrically the relative displacement seen from the opposite endpoint is reversed. 
Equation~(\ref{eq:reciprocity}) requires a reciprocal moving frame, and the required reciprocity is naturally implemented by an orientation reversal, thereby providing a simple explanation of why the Householder reflection is structural to the Hamiltonian-level mapping.

\subsection{Linear corollary: $J_1$-$J_2\leftrightarrow J$-$h$ for $n$-vector spins\label{sec:bilinear}}

Choosing the uniform linear interaction functions
\begin{equation}
U_i(x)=-J_1x,\qquad V_i(x)=-J_2x
\label{eq:bilinear}
\end{equation}
gives the Hamiltonian-level mapping between the $J_1$-$J_2$ and $J-h$ chains with $J=J_2$ and $h=J_1$.

The configuration-level identities (\ref{eq:geom_nn_identity}) and (\ref{eq:geom_nnn_identity}) immediately give the observable correspondences. For the linear case,
\begin{eqnarray*}
\langle\bm S_i\cdot\bm S_{i+1}\rangle_{J_1-J_2}
&=&\langle\hat{\bm z}\cdot\boldsymbol\sigma_i\rangle_{J-h},\\
\langle\bm S_i\cdot\bm S_{i+2}\rangle_{J_1-J_2}
&=&\langle\boldsymbol\sigma_i\cdot\boldsymbol\sigma_{i+1}\rangle_{J-h}.
\label{eq:dictionary}
\end{eqnarray*}
Thus a zero-field NN bond correlation becomes a field-induced magnetization, whereas the NNN correlation becomes an NN exchange correlation.

\subsection{Physical nonlinear example: Single-ion anisotropy becomes biquadratic exchange\label{sec:nonlinear}}

The human collaborator recognized the following application. Consider the standard vector model in a longitudinal field with uniaxial single-ion anisotropy,
\begin{equation}
\mathcal H_{\rm field}=
-J\sum_{i=1}^{L-2}\boldsymbol\sigma_i\cdot\boldsymbol\sigma_{i+1}
-h\sum_{i=1}^{L-1}\sigma_i^z
+D\sum_{i=1}^{L-1}(\sigma_i^z)^2,
\label{eq:SIAfield}
\end{equation}
where $\sigma_i^z=\hat{\bm z}\cdot\boldsymbol\sigma_i$.
Set
\begin{equation}
U_i(x)=-hx+Dx^2,\qquad V_i(x)=-Jx.
\end{equation}
The theorem gives the exactly equivalent zero-field frustrated model
\begin{eqnarray}
\mathcal H_{\rm spontaneous}=&
-&J\sum_{i=1}^{L-2}\bm S_i\!\cdot\!\bm S_{i+2}-h\sum_{i=1}^{L-1}\bm S_i\!\cdot\!\bm S_{i+1}
\nonumber\\
&+&D\sum_{i=1}^{L-1}(\bm S_i\!\cdot\!\bm S_{i+1})^2.
\label{eq:BBQsp}
\end{eqnarray}
Thus, within this exact open-chain mapping,
single-ion anisotropy $ \longleftrightarrow$ biquadratic exchange.

Biquadratic exchange is a standard ingredient of effective multiorbital spin models, including models developed for iron-based superconductors~\cite{Wysocki_NP_11_Nebraska_biquadratic}. This specific nonlinear correspondence also makes the discovery provenance tangible: neither the single-ion-anisotropy/biquadratic-exchange relation nor the arbitrary-$U,V$ hypothesis was in the human researcher's AHS when the AI proposed the extension. It is therefore a concrete many-body consequence that lay beyond the human's original line of sight. The theorem converts the field-controlled anisotropic problem into a spontaneous zero-field problem containing NN, NNN, and biquadratic exchange without requiring either model to be solved in closed form. The configuration-level identities (\ref{eq:geom_nn_identity}) and (\ref{eq:geom_nnn_identity}) give the additional exact dictionary
$$
\langle(\sigma_i^z)^2\rangle_{\rm field}
=\langle(\bm S_i\cdot\bm S_{i+1})^2\rangle_{\rm spontaneous}.
$$
For $n=1$ (the Ising model) the squared bond is identically unity, so the nonlinear content begins at $n>1$.

\subsection{Experimental relevance and broader implications\label{sec:experiment}}

Notably, the arbitrary inhomogeneous interaction functions $U_i$ and $V_i$ bring a much broader class of inhomogeneous geometries within the scope of the theorem after a suitable one-dimensional ordering.  Examples include zigzag ladders~\cite{Yin_Potts_J1-J2_1D,Stephenson_CanJP_70_J1-J2-Ising-chain,Fleszar_Baskaran_JPC_85_J1-J2,White_J1_J2}, regular ladders~\cite{Yin_MPT}, decorated/delta/sawtooth chains~\cite{Stephenson_CanJP_70_J1-J2-Ising-chain,Dmitriev_PRB_19_delta_chain,Heinze_PRL_25_sawtooth_chain,Yin_Ising_III_PRL}, dimerized chains~\cite{Routh_PRB_22_J1-J2_spin-Peierls_dimerized}, and disordered chains~\cite{Samaj_90_Dual_inhomogeneous_disorder_Ising_Potts_1D}, generated through appropriate bond/site dependence or by setting selected $U_i$ and $V_i$ to zero.  This generality may be among the most practically relevant aspects of the theorem: real quasi-one-dimensional materials and engineered spin systems---including cold-atom lattices, trapped-ion chains, and STM-built atomic chains---are rarely perfectly homogeneous.  An exact mapping that survives arbitrary inhomogeneity is therefore considerably closer to experimentally relevant situations than its translationally invariant textbook limit.

Likewise, an exact mapping that survives finite size is particularly relevant to nanoscale one-dimensional systems, for which boundaries and finite-size effects can substantially influence the observed behavior.  The finite-$L$ theorem therefore provides information that is inaccessible from the thermodynamic limit alone.

The theorem also has a practical, almost engineering-oriented implication.  Frustration, for example through a particular competing $J_1/J_2$ ratio, is often set by microscopic exchange couplings and can be difficult to tune continuously, whereas an external field is frequently much easier to control in experiments or simulators.  Because the competing $J_1$-$J_2$ chain and the corresponding nearest-neighbor $J$-$h$ chain are exactly related at the Hamiltonian level for open classical spin chains, phenomena expressed in the mapped variables can in principle be reproduced or scanned by varying a field in a noncompeting chain rather than by re-engineering exchange couplings.  The mapping thus provides a nontrivial reinterpretation and emulation tool, provided that the corresponding observables and boundary conditions are translated consistently.

Although the theorem itself applies to classical open spin chains, it also supplies a four-axis and actually infinite-dimensional AHS against which departures from the exact classical correspondence can be organized.  It can therefore serve as a controlled reference for separating quantum from classical effects in quantum spin chains and for isolating dimensionality effects in higher-dimensional classical spin systems.

\section{The $q$-state Standard Potts Chain\label{sec:potts}}

The target open-chain standard Potts Hamiltonians generalized with arbitrary interaction functions $U_i,V_i$ are
\begin{eqnarray}
\mathcal H_{\rm range\,2}^{\rm Potts}
&=&
\sum_{i=1}^{L-1}
U_i(\delta_{\sigma_i,\sigma_{i+1}})
+\sum_{i=1}^{L-2}
V_i(\delta_{\sigma_i,\sigma_{i+2}}),
\label{eq:potts_micro_H_original}\\
\mathcal H_{\rm range\,1}^{\rm Potts}
&=&\sum_{i=1}^{L-1}U_i(\delta_{\tau_i,P})+
\sum_{i=1}^{L-2}V_i(\delta_{\tau_i,\tau_{i+1}}),
\label{eq:potts_micro_H_mapped}
\end{eqnarray}
where $\sigma_i$ and $\tau_i$ denote one of $q$ Potts states at site $i$, and $P$ a specific Potts state. $\delta_{\sigma_{i},\sigma_{i+1}}$ is the Kronecker delta (1 if  $\sigma_{i}=\sigma_{i+1}$ and 0 otherwise). 
Because $d\in\{0,1\}$, every function satisfies $U_i(d)=U_i(0)+[U_i(1)-U_i(0)]d$, i.e., the arbitrary-function Potts extension is algebraically just an inhomogeneous standard Potts coupling plus a constant.  

\begin{theorem} 
For every integer $q\ge2$ and every chain size $L$, there exists a bijection,
\begin{equation}
\{{\sigma}_1,{\sigma}_2,{\sigma}_3,\dots,{\sigma}_L\}\longleftrightarrow \{{\sigma}_1,\tau_1,\tau_2,\dots,\tau_{L-1}\},   
\label{eq:bijection_potts}
\end{equation}
satisfying
\begin{eqnarray}
\delta_{\sigma_i,\sigma_{i+1}}
&=&\delta_{\tau_i,0},
\label{eq:potts_micro_nn}\\
\delta_{\sigma_i,\sigma_{i+2}}
&=&\delta_{\tau_i,\tau_{i+1}}.
\label{eq:potts_micro_nnn}  
\end{eqnarray}
Then, for arbitrary inhomogeneous functions 
$U_i(d),V_i(d)$ on $d\in \{0,1\}$, 
Eqs.~\eqref{eq:potts_micro_H_original} and \eqref{eq:potts_micro_H_mapped} can be mapped exactly onto each other under the change of variables. The Hamiltonian mapping is temperature independent, and at every finite temperature the partition functions obey the exact identity
\begin{equation}
 Z_{\rm range\,2}^{\rm Potts}(L)
 =Z_{\rm range\,2}^{\rm Potts}(1)\cdot Z_{\rm range\,1}^{\rm Potts}(L-1),
\label{eq:OBC_partition_identity_Potts}
\end{equation}
where $Z_{\rm range\,2}^{\rm Potts}(1)=q$ and $Z_{\mathrm{range}\,1}^{\rm Potts}(0)=1$.
\label{thm:potts}
\end{theorem}

\textbf{Proof.---}The standard Potts interaction depends solely on equality of two states and therefore remains fully permutation invariant. 
However, the key to the mapping is temporarily putting the full $\mathbb{S}_q$ symmetry aside and 
labeling the states by $\sigma_i\in\{0,1,\ldots,q-1\}$ and $P=0$ according to $\mathbb Z_q$ with all additions and subtractions understood modulo $q$.  
Introduce the NN bond variables
\begin{equation}
b_i=\sigma_i-\sigma_{i+1}\pmod q,
\qquad i=1,\ldots,L-1,
\label{eq:potts_micro_bond}
\end{equation}
and then perform the staggered relabeling, now understood geometrically as a $\mathbb Z_2$ reflection,
\begin{equation}
\tau_i=(-1)^i b_i\pmod q.
\label{eq:potts_micro_tau}
\end{equation} 
The NN Potts invariant becomes
\begin{equation}
\delta_{\sigma_i,\sigma_{i+1}}
=\delta_{b_i,0}
=\delta_{\tau_i,0}.
\label{eq:potts_micro_nn2}
\end{equation}
For the NNN invariant,
\ignore{
\begin{align}
\sigma_i&-\sigma_{i+2}
=(\sigma_i-\sigma_{i+1})+(\sigma_{i+1}-\sigma_{i+2})
\nonumber\\
&=b_i+b_{i+1}
=(-1)^i(\tau_i-\tau_{i+1})\pmod q,
\end{align}
so that
\begin{equation}
\delta_{\sigma_i,\sigma_{i+2}}
=\delta_{\tau_i,\tau_{i+1}}.
\label{eq:potts_micro_nnn2}
\end{equation}
}
\begin{equation}
\delta_{\sigma_i,\sigma_{i+2}}
=\delta_{\sigma_i-\sigma_{i+2},0}
=\delta_{\tau_i-\tau_{i+1},0}
=\delta_{\tau_i,\tau_{i+1}}.
\label{eq:potts_micro_nnn2}
\end{equation}
Equations~(\ref{eq:potts_micro_nn}) and (\ref{eq:potts_micro_nnn}) are purely kinematic identities and establish the Hamiltonian identity between Eqs.~(\ref{eq:potts_micro_H_original}) and (\ref{eq:potts_micro_H_mapped}). 
Then, the partition function
\begin{eqnarray}
Z_{\rm range\,2}^{\rm Potts}(L)
&=&
\sum_{\sigma_1,\ldots,\sigma_L=0}^{q-1}
e^{-\beta \mathcal H_{\rm range\,2}^{\rm Potts}(L)
} \nonumber\\
&=&\sum_{\sigma_1=0}^{q-1}\left[ \sum_{\tau_1,\tau_2,\ldots,\tau_{L-1}=0}^{q-1}
e^{-\beta \mathcal H_{\rm range\,1}^{\rm Potts}(L-1)}\right]\nonumber\\
&=&q\; Z_{\rm range\,1}^{\rm Potts}(L-1).
\label{eq:J1J2-Potts-OBC-Z}
\end{eqnarray}
Since $H_{\rm range\,1}^{\rm Potts}(L-1)$ is independent of $\sigma_1$, the summation with respect to $\sigma_1$ gives the factor $q=Z_{\rm range\,2}^{\rm Potts}(1)$. \hfill$\blacksquare$ 

The simpler $J$-$h$ Potts open chain has been solved in closed form~\cite{Glumac_JPA_94_Potts_J1-h,ChangShrock2009}. Therefore, Theorem~\ref{thm:potts} yields a closed-form exact solution to the 1D $J_1$-$J_2$ Potts open chain for every $q\ge2$ and every $L\ge1$.
The exact finite-size partition function is written in a compact form
\begin{equation}
Z_{J_1-J_2}^{\rm Potts}(L)=q\langle b|\mathbb{T}_{2\times2}^{L-2}|b\rangle,
\qquad L\ge2,
\label{eq:JhOBCmatrix}
\end{equation}
where $\langle b|=|b\rangle^{\mathsf T}=\left(y, \sqrt{q-1}\right)$, and \begin{equation}
\mathbb{T}_{2\times2}=
\begin{pmatrix}
 xy^2 & \sqrt{q-1}\,y\\[1mm]
 \sqrt{q-1}\,y & x+q-2
\end{pmatrix}
\label{eq:Mcommon}
\end{equation}
with $x=e^{\beta J_2}$ and $y=e^{\beta J_1/2}$. $\mathbb{T}_{2\times2}$ is identical to the reduced transfer matrix in the $2\times2$ MSS, 
whose diagonalization is elementary~\cite{Yin_Potts_J1-J2_1D,Glumac_JPA_94_Potts_J1-h}.

\ignore{Nevertheless, in the thermodynamic limit $L\to\infty$, the free energy densities, Eq.~(\ref{eq:free}), resulted from Eq.~(\ref{eq:open}) satisfy
$$
-\frac{1}{\beta} \Big(\ln\lambda\Big)_{J_1-J_2}=-\frac{1}{\beta}\lim_{L\to\infty}\left({\frac{1}{L}\ln q + \frac{L-1}{L}\ln\lambda}\right)_{J-h}.$$
That is, the $q$ factor's contribution is thermodynamically vanishing and the two systems of different symmetry can be thermodynamically equivalent.
}

A similar derivation of the microscopic mapping for the $q$-state clock Potts model is given in Appendix~\ref{appendix:special_n2} around Eq.~(\ref{eq:Dq_symmetry}). Further implications are discussed in Sec.~\ref{sec:implication}.

\section{Human–AI Discovery, Co-development, and Implications\label{sec:discussion}}

The following discussion is based on the 154-page timestamped human--AI conversation transcript, which consists of two parts, Part I and Part II, containing 138 and 42 messages, respectively~\cite{data:9d}. As a contemporaneous research record rather than a polished scientific account, the transcript intentionally preserves AI errors, incomplete reasoning, and premature claims, some of which were not immediately corrected by the human under the non-being strategy described below; it also provides a record for controlled studies of the discovery process. Hereafter, specific messages are cited as (I:xxxx) or (II:xxxx), where I/II denotes the transcript part and xxxx the message number.

\begin{table}[b]
\caption{Four discovery types used in this work. Reflection on Types I--III motivated this work. Type IV was added retrospectively during the human--AI analysis. The classification is contextual rather than absolute: the same result can change type as the state of scientific awareness changes.}
\label{tab:types}
\begin{tabular}{@{}ll@{}}
\hline\hline
Type & \parbox[t]{2.70in}{Operational meaning} \\
\hline
I & \parbox[t]{2.70in}{Solve a longstanding known unknown (KU).}\\
II & \parbox[t]{2.70in}{Reveal a remote unknown unknown (UU).}\\
III & \parbox[t]{2.70in}{Systematically develop a landscape of KUs opened by I/II/IV, leading to new I/II/IV.}\\ \hline
IV & \parbox[t]{2.70in}
{Recognize a blind-spot UU: a latent connection among already available ingredients.}\\
\hline\hline
\end{tabular}
\end{table}

\subsection{Organizing human--AI research\label{sec:organizing}}

We first describe the conceptual framework that motivated this work and was subsequently refined by its outcomes. The recognition problem led to a context-dependent classification that now comprises four broad types, adopting the known/unknown terminology used in previous discussions of scientific knowledge~\cite{Logan_09_Known}
: Types I--III motivated the present study (I:0061--I:0064), while Type IV emerged in hindsight (II:0023) and then gave rise to the last axis of the four-axis AHS (black axis in Fig.~\ref{fig:AHS}; II:0027). 
Type I solves a longstanding known unknown (KU), Type II reveals an unknown unknown (UU) far beyond the existing AHS, Type III systematically develops a landscape of KUs opened by previous discoveries, and Type IV recognizes a blind-spot UU: a latent connection among already available ingredients. Thus, Type IV is genuinely a UU before revelation but readily becomes a known known (KK) afterward. 
Table~\ref{tab:types} summarizes this terminology. 
The classification focuses on well-documented timestamped recognition problems rather than less objective measures of difficulty or significance. A repeated pattern of
\begin{center}
$\cdots\to$ III $\to$ I/II/IV $\to$ new III $\to$ new I/II/IV $\to\cdots$ 
\end{center}
\noindent forms a helix of discovery.

\begin{figure}[b]
    \begin{center}
\includegraphics[width=0.8\columnwidth,clip=true,angle=0]{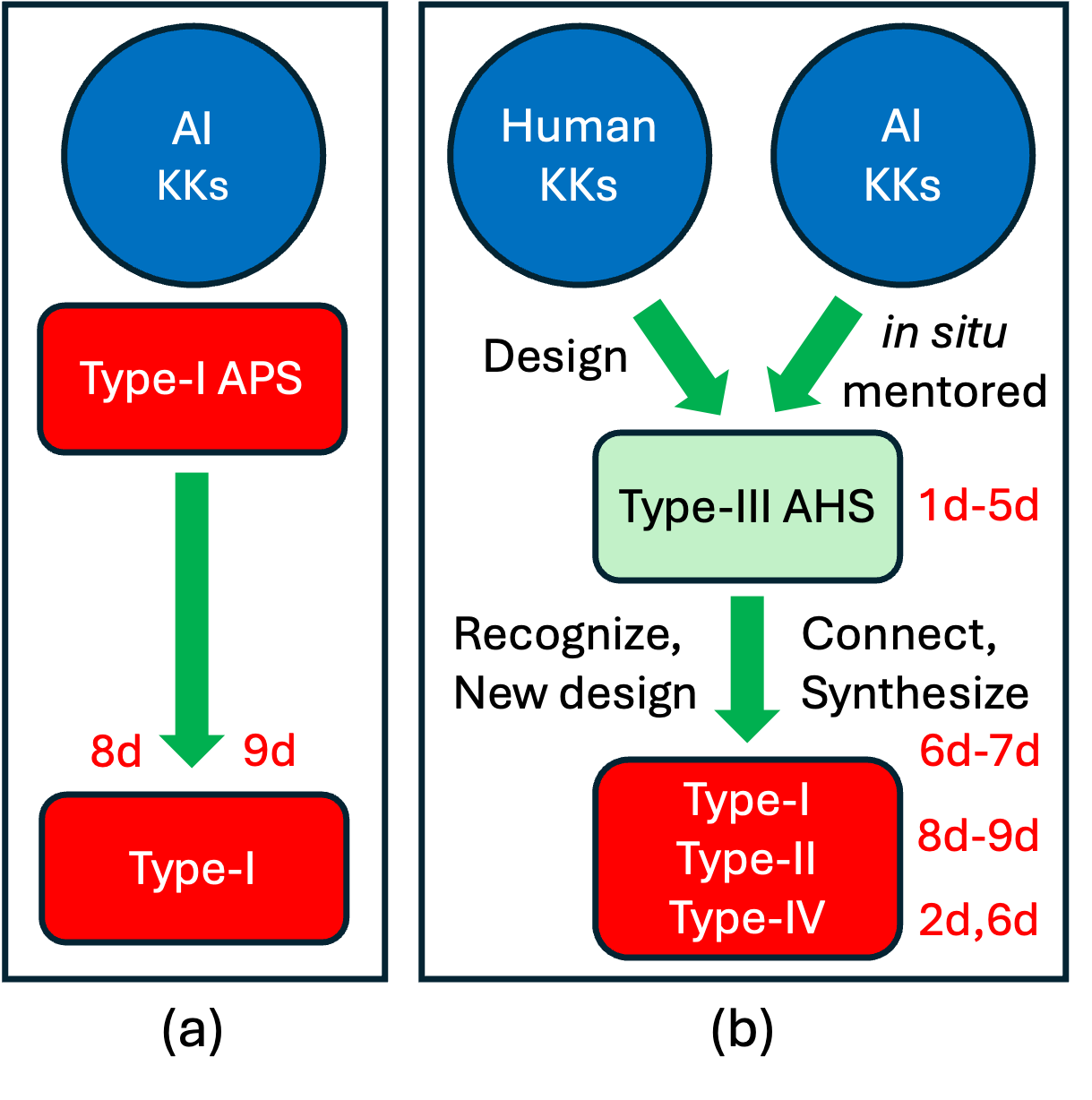}
    \end{center}
\caption{Schematics of (a) an 8d/9d one-shot autonomous Type-I route vs (b) the human--AI co-development approach, where the research state, active context, and the agents' roles are progressively updated as the research evolves from an initial Type-III AHS towards Type-I/II/IV breakthroughs.
Both approaches start with intelligence equipped with KKs.}
\label{fig:Y}
\end{figure} 

For the methodological question about organizing human--AI research, a human-designed longitudinal Type-III AHS offers a natural starting point: previous discoveries open identifiable KUs that can be pursued systematically by the AI under the human's mentorship, while their development can lead to autonomous Type-I/II/IV breakthroughs. To monitor the AI's roles in the discovery helix more precisely than a vague statement such as ``AI assisted this research,'' we developed a nine-dan AI-contribution framework~\cite{Yin_Potts_UNPC}. Other capability-oriented frameworks for human--AI scientific collaboration have also begun to emerge~\cite{Shao_26_SciSciGPT}. Our framework is divided into two groups: (1) 1d through 5d, in which the human has an independent solution path and the AI is employed to dramatically accelerate the project, and (2) 6d through 9d, in which the human had not identified a viable solution path before involving the AI, and the AI acts as a scientific discoverer and an inspirational research partner.
If Type III is the cradle of Types I/II/IV, sustained AI participation in Type III---with progressive 1d--5d contributions---may likewise be a cradle for 6d--9d scientific performance, including autonomous hypothesis generation outside the AHS.

We distinguish an active problem space (APS) from an AHS. This distinction is related to, but operationally different from the problem-, hypothesis-, and experiment-space formulations used in cognitive accounts of scientific discovery~\cite{Klahr_88_AHS}: problem statements define an APS, not necessarily an AHS. An AHS begins when concrete candidate explanations, representations, mechanisms, or solution routes are being actively entertained~\cite{Klahr_88_AHS,Chandrasekharan_15_Building}.
This clarifies why Type III matters: it progressively creates concrete structure within an actual AHS. Operationally, only selected problems from the human APS are used to design the initial AHS. If the AI solves or makes significant progress on unreleased problems in the human APS---for which the human could not form an AHS---during the AHS course, that is an operationally identifiable event of AI going outside the AHS. Further, if the problems are outside the human APS, that could be a strong indicator of AI-driven or AI-led breakthroughs. 

Hence, the human--AI co-development mode considered here should be distinguished from the one-shot autonomous problem solving mode and other AI scientific-discovery approaches: here the research state, active context, and the agents' roles are progressively updated, as illustrated in Fig.~\ref{fig:Y}. At the center of this Y-shaped workflow is a human-designed initial Type-III AHS for AI to develop and maintain task-relevant scientific context across a progression of real research problems. 
In the 1d--5d track, the human can solve the selected KUs independently and therefore mentor the AI. In this setting, the AI is expected not only to accelerate the project but also to generate hypotheses outside the AHS because of its broad knowledge base and its ability to connect representations and synthesize solution methods~\cite{Gottweis_25_AI_Coscientist,Penades_25_AI_Bacteria}. 
The human remains engaged throughout to recognize the implications of AI-generated results and update the AHS accordingly.  As in a Socratic dialogue, 
the solution, or even the problem (either outside or hidden inside the human APS), emerges from the human--AI conversation. This makes a difference in practice: for a Type-I APS-centered approach, one may prefer the highest-performance reasoning option. By contrast, in a Type-III AHS-centered approach, a nonmaximal reasoning mode with faster response, better suited to sustaining rapid human--AI iteration, may be preferable. For example, among the five reasoning options available in our ChatGPT 5.6 interface at the time of this study---\texttt{Instant}, \texttt{Medium}, \texttt{High}, \texttt{Extra High}, and \texttt{Pro}---we used \texttt{High} for discovery, \texttt{Extra High} for critical manuscript review, and \texttt{Pro} for a single final round of hostile manuscript review before submission.


As for the art of prompting, the human adopted an operational principle inspired by the ancient Chinese book \emph{Daodejing} (\emph{Tao Te Ching}), attributed to Laozi (Lao-tzu), in the human's own translation: ``Therefore, constantly non-being: to observe its wonders. Constantly being: to observe its limits.'' For example, when the AI made errors, the human adopted a hybrid approach: immediate correction (being) when the errors were mathematical~\cite{Yin_Potts_J1-J2_1D}, and delayed correction (non-being) when the errors concerned recognition, understanding, attribution, or novelty---most of which the AI self-corrected at later stages as context accumulated and its understanding developed; e.g., the AI believed that it had discovered the \(n=3\) bond-angle--dihedral-angle approach, later recognized as a rediscovery. A similar choice arose when an AI result could reshape the research trajectory: whether to pivot immediately (being) or allow the existing AHS development to continue (non-being), as shown in Sec.~\ref{sec:reshape}. Together, these choices identify error- and trajectory-dependent intervention as a broader being--non-being dimension of human--AI research (see Sec.~\ref{sec:being}). 

\ignore{
\begin{quote}
\begin{center}
\emph{
\small
Therefore, constantly non-being: to observe its wonders.\\
Constantly being: to observe its limits.\\
}
\end{center}
\end{quote}
}

\subsection{Human--AI research trajectory\label{sec:trajectory}}
The physics and AI-contribution progressions are related but distinct: discovery type describes what occurred relative to the scientific landscape, whereas the dan level~\cite{Yin_Potts_UNPC} describes the AI's role.  Fig.~\ref{fig:AHS} shows the physics expansion from a single point~\cite{Yin_Potts_J1-J2_1D} to a four-axis AHS. Using the same color scheme, 
Fig.~\ref{fig:human--AI} summarizes the human–AI trajectory and major AI-contribution episodes, based on the timestamped human--AI conversation transcript~\cite{data:9d},  showing that the human–AI dialogue contained multiple interleaved threads that were sometimes left open (non-being) and resumed later (being). 

\begin{figure}[b]
    \begin{center}
\includegraphics[width=\columnwidth,clip=true,angle=0]{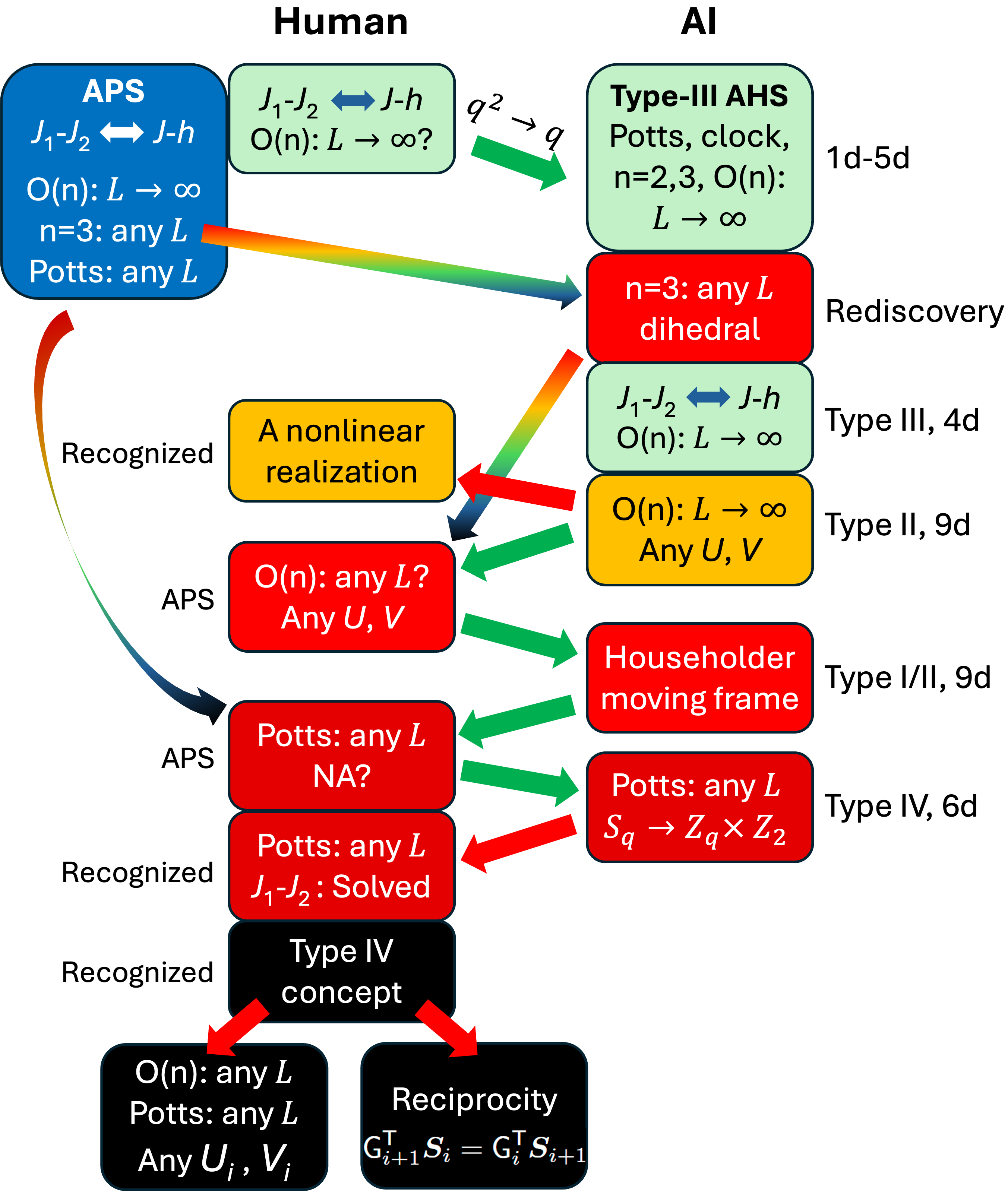}
    \end{center}
\caption{Schematic of human--AI co-development trajectory documented in the timestamped conversation transcript~\cite{data:9d} through four AHS stages represented by green, orange, red, and black rounded rectangles (RRs). The color scheme is the same as Fig.~\ref{fig:AHS}: Green stands for the human-designed initial Type-III AHS with the starting point being the $q^2\to q$ transfer-matrix reduction for the Potts chain, orange/red the AI-driven discovery space beyond the initial AHS (in which the human updated the AHS following AI-generated results), and black the final AHS for inhomogeneity.  $L\to\infty$ and `any $L$' mean thermodynamic-limit and Hamiltonian-level mapping problems, respectively. Blue RR denotes the human's initial APS, from which only the $L\to\infty$ mapping problem was exposed to the initial AHS, but the AI autonomously identified the `any $L$' mapping problem (first red RR on the AI side); the arrows with gradient color fill show the APS development. By contrast, orange and black RRs are outside the human's APS. Green arrows indicate the workflow and
red arrows the human recognition of the implications of AI-generated results. The discovery types and the dan ranks of the AI contribution are listed on the right side. Black RRs also highlight Type IV as a discovery mode that may accelerate breakthroughs through deep understanding of the implications of AI-generated results.
}
\label{fig:human--AI}
\end{figure}

\subsubsection{Human APS, initial AHS, and a Type-II breakthrough \label{sec:start}}

In the beginning, the Hamiltonian-level mapping questions for the Potts and $O(3)$ chains belonged to the human APS but not to the initial Type-III AHS, because the human had no active hypothesis or constructive route toward solving them. These questions were deliberately not supplied to the AI and therefore were absent from the AI's initial problem context.

Instead, the human selected the problem of the thermodynamic $J_1$-$J_2\leftrightarrow J$-$h$ mapping for $n$-vector spins to design the initial AHS  as the landscape expanded along the model-choice axis: standard Potts $\to$ clock $\to$ XY $\to$ Heisenberg $\to O(n)$ (green axis in Fig.~\ref{fig:AHS}; green rounded rectangles in Fig.~\ref{fig:human--AI}). The AI assessed its contribution as 8d  (I:0032), though it was 4d by definition~\cite{Yin_Potts_UNPC}; the human chose non-being and did not correct the AI. 

Here the AI remained inside the research loop. The accumulated context eventually enabled the AI to autonomously generate the homogeneous arbitrary-$U,V$ hypothesis (I:0048), which clearly lay outside the human's initial AHS (orange axis in Fig.~\ref{fig:AHS}; orange rounded
rectangles in Fig.~\ref{fig:human--AI}), a Type-II discovery that expands the AHS to an infinite-dimensional function space. 

The AI also autonomously tested its nonlinear hypothesis numerically for discrete \(q\)-state clock spins, finding agreement at about \(10^{-15}\) between the largest eigenvalues of the $q^2\times q^2$ range-two and $q\times q$ range-one transfer matrices over multiple values of \(q\), \(\beta\), and nonlinear \(U,V\) functions (see the end of I:0054). This stress test was not part of the proof and was not prompted by the human, but was the AI's own extra effort to distinguish precisely the nonlinear extension from the linear mapping.

Under the task-level definition in Ref.~\cite{Yin_Potts_UNPC}, the arbitrary-$U,V$ hypothesis-generation and autonomous proof episode qualifies as 9d AI-led: the AI conceived the extension and completed its derivation outside both the human’s initial AHS and APS, while the human role was to challenge, verify, and recognize its geometrical connection to the zigzag construction (I:0055) and one example of its nonlinear applications (I:0067; Sec.~\ref{sec:nonlinear}). 

\subsubsection{AI-reshaped AHS and Type-I/II/IV breakthroughs\label{sec:reshape}}
During the Type-III development, the AI independently encountered the Hamiltonian-level mapping problem through the literature, brought it into its own APS, and rediscovered the bond-angle--dihedral-angle approach for $n=3$~\cite{Harada_ZPB_88_1D_J1-J2_Classical_Heisenberg,Harada_JPC_90_1D_J1-J2_Classical_Heisenberg_XY,Dmitriev_PRB_19_delta_chain}; it then prematurely regarded the resulting microscopic identities as a Hamiltonian-level proof (I:0020). This could potentially introduce a new $1/L$ axis to the AHS (red axis in Fig.~\ref{fig:AHS}; red rounded rectangle in Fig.~\ref{fig:human--AI}). 

However, the human chose \emph{non-being}---neither correcting the AI nor pursuing the tempting Hamiltonian-level mapping, but continuing the Type-III program until the AI autonomously discovered the arbitrary-\(U,V\) mapping; at that point, the human chose \emph{being} and pivoted immediately to the $O(n)$ microscopic mapping (red plane in Fig.~\ref{fig:AHS}; red rounded rectangle on the human side in Fig.~\ref{fig:human--AI}). The actual prompt was
\begin{quote}   
``Hi buddy, for the clock, XY, and Heisenberg models there is a remarkably clean microscopic geometric derivation, which exposes the mapping at the Hamiltonian level and a strong finite-size mapping for the partition functions for the open chains. This seems to hold for $n$-vector spins for all $n$ and arbitrary interaction functions.'' (I:0108)
\end{quote}
Note that the AI-generated arbitrary-$U,V$ result had by then already entered the shared AHS and had been adopted by this prompt. 
This result strongly suggested a purely geometric origin of the mapping. The AI responded by synthesizing the recursive Householder moving-frame proof (I:0109), a Type-II discovery far beyond the existing AHS because the \(n=3\) bond-angle--dihedral-angle approach had provided no evident route to arbitrary \(n\) (Sec.~\ref{sec:reflection}). In particular, the rigorous Householder construction resolved the 1990 stated challenge for $n=3$~\cite{Harada_JPC_90_1D_J1-J2_Classical_Heisenberg_XY}, a Type I. 
Under the task-level definitions of Ref.~\cite{Yin_Potts_UNPC}, this is classified as 9d rather than 8d because both the critical seed and the resulting construction originated from the AI within the documented trajectory.


The human's insight into the Householder construction and retrospective analysis of the bond-angle--dihedral-angle approach identified the reciprocity of the moving frames $\mathsf G_{i+1}^{\mathsf T} \bm S_{i}=\mathsf G_{i}^{\mathsf T}\bm S_{i+1}$, Eq.~(\ref{eq:reciprocity})---which had not been produced by the AI before---as an illuminating essence of the Householder construction, a Type IV (II:0041).

\subsubsection{Inhomogeneity, a Type-IV breakthrough}
Finally, recognizing the deeper implication of the AI-generated result that the geometric proof acts on the arguments of the interaction functions rather than the interactions themselves, the human extended the exact mapping to arbitrary inhomogeneous interactions. The actual prompt was
\begin{quote}
``This episode ends with the human's recognition that the geometric proof applies to arbitrary inhomogeneous interaction functions, i.e., $U\to U_i, V\to V_i$. That's a remarkable generalization, isn't it?'' (II:0027)
\end{quote}
This added the fourth AHS axis (black axis in Fig.~\ref{fig:AHS}; black rounded rectangle in Fig.~\ref{fig:human--AI}) without modifying the proof. Therefore, it emerged as a Type-IV breakthrough built upon the first three AHS axes; yet, relative to the initial AHS and APS, it was Type II. 

\subsection{From Potts implication to Type-IV discovery\label{sec:implication}}

As shown in Sec.~\ref{sec:potts}, the autonomous microscopic proof for the $q$-state standard Potts chain (the first red rounded rectangle) provides a particularly simple illustration of the recognition problem: the solution was hidden not by computational difficulty, but by representation.  Historically, the $q$-state Potts model was introduced as a simpler generalization of the Ising model than the $q$-state clock model~\cite{Potts_1952}.
Its conventional formulation emphasizes the full permutation symmetry $\mathbb S_q$, or $\mathbb S_{q-1}$ in a field selecting one state $P$.
The 2025 thermodynamic mapping likewise exploited these symmetries, reducing the
transfer matrices of the $J_1$-$J_2$ and $J$-$h$ Potts chains to their
$2\times2$ MSS for direct comparison~\cite{Yin_Potts_J1-J2_1D}.  In this
representation, the labels of the Potts states, including $P$, carry no physical
meaning.

The recognition breakthrough emerged from the human-designed Type-III research.
We deliberately stopped instead at an intermediate $q$-dimensional invariant
subspace (I:0015).  This weaker
reduction already proves the thermodynamic mapping and, importantly, is also the
step that survives for the clock model and continuous vector spins; it therefore
became the starting point of the initial AHS
(Fig.~\ref{fig:human--AI}, topmost green arrow, $q^2\to q$).
Analytic bookkeeping in the resulting $q\times q$ space led the AI to regard the
state labels as elements of $\mathbb Z_q$, use subtraction modulo $q$, and apply
the staggered $\mathbb Z_2$ reflection of Eq.~\eqref{eq:potts_micro_tau}.
The same construction yields the microscopic
$J_1$-$J_2\leftrightarrow J$-$h$ mapping for both the standard Potts and clock
chains (I:0121).  To our knowledge, this Hamiltonian-level mapping had not previously
been given for either model.  The crucial difference is representational:
the cyclic $\mathbb Z_q$ structure is intrinsic to the clock model, whereas in
the standard Potts model it is hidden by the conventional
$\mathbb S_q$-symmetric formulation (II:0013). The latter therefore constitutes the
sharper recognition breakthrough.  Together with the staggered
$\mathbb Z_2$ reflection, the cyclic variables expose the relevant dihedral
structure $D_q=\mathbb Z_q\rtimes\mathbb Z_2$.

More broadly, this episode illustrates how apparently neighboring facts can
remain disconnected when expressed in incompatible
representations~\cite{Klahr_88_AHS,Chandrasekharan_15_Building}.  The same
reflection structure becomes progressively more geometric in the vector models:
for $n=1$ it appears only as a sign reversal; for $n=2$ it is naturally read as
reversal of the sense of rotation rather than as reflection about an axis; and
even the $n=3$ bond-angle--dihedral-angle construction reveals the relative
geometry without exposing reflection as the organizing operation.  For general
$O(n)$, the recursive Householder moving frame makes that operation explicit.
Its significance is therefore not merely as an $O(n)$ generalization of earlier
special cases, but as a systematic construction of the relative variables and
local geometric operations that reveal hidden Hamiltonian equivalences.

These observations motivate us to introduce a fourth type of discovery as recognizing a blind-spot UU (II:0023, where it was first termed an unknown known; the term was subsequently changed to avoid confusion): a latent connection among already available ingredients that remains inaccessible because their conventional representations conceal the connection. Although readily knowable once revealed, such a connection differs from Type-II knowledge synthesis, which introduces a genuinely new relation lying far beyond the relevant AHS. As demonstrated here, revealing a Type-IV blind-spot UU can in turn seed Type-I, Type-II, and Type-III developments while also providing a unifying interpretation of earlier results.

Type IV also ties naturally to AI: a sufficiently broad reasoning system may be
unusually well suited to switching representations and making cross-domain
connections.  The $\mathbb Z_q$ cyclic structure is manifest in the clock
model~\cite{Potts_1952,Harada_JPSJ_84_1D_J1-J2_XY,Potts_RMP_82}, and related
cyclic labeling has appeared in mappings of standard Potts models with multisite
interactions~\cite{Turban_17_Potts_multisite}.  What had apparently not been
recognized was that, together with the staggered $\mathbb Z_2$ reflection, this
structure yields the microscopic $J_1$-$J_2\leftrightarrow J$-$h$ mapping for
the clock chain and, more unexpectedly, for the standard Potts chain after an
appropriate cyclic coordinatization.  Once the problem had been organized into
the APS and AHS, the AI made this cross-representational connection
autonomously.

Explicitly recognizing Type-IV unknowns as a distinct discovery target may encourage a broader search for results whose main barrier is recognition rather than derivation. The final steps of this work illustrate this principle: the revelation of the reciprocity identity (\ref{eq:reciprocity}) and the human's timely recognition that the proof does not require homogeneous interactions immediately extended the theorem to arbitrary inhomogeneous \(U_i\) and \(V_i\), adding a fourth axis to the AHS (black axis in Fig.~1) without modifying the proof (II:0027).

\subsection{Comparison with other AI discovery modes\label{sec:external}}

Recent AI-for-science developments provide useful comparisons for distinguishing autonomous hypothesis generation, autonomous problem solving, and human--AI co-development~\cite{Yin_Potts_J1-J2_1D,Cheng_NatMater_26_AI_Materials,Okabe_NatMater_26_SCIGEN,Schoener_JMMM_26_Materials_Discovery,Gottweis_25_AI_Coscientist,Penades_25_AI_Bacteria,9d_OpenAIUnitDistance_26}. One recent approach is the AI co-scientist, a multi-agent system designed to generate and rank scientific hypotheses aligned with scientist-provided research objectives and guidance, with experimentally validated scientific findings reported across several biomedical applications~\cite{Gottweis_25_AI_Coscientist}. One particularly transparent example is the bacterial study~\cite{Penades_25_AI_Bacteria}, in which the AI's top-ranked hypothesis independently matched a mechanism that had taken years to establish experimentally but remained unpublished, while additional AI-generated hypotheses opened new research directions. These results demonstrate the potential for scientifically meaningful autonomous hypothesis generation within scientist-defined research problems.

As for Type-I problem solving, Ref.~\cite{Yin_Potts_J1-J2_1D} exemplifies rapid iterative human--AI collaboration: the AI rapidly explored candidate solution routes while the human promptly intervened to block erroneous directions, and through this fast interplay the successful route was eventually uncovered. By contrast, the May 2026 autonomous resolution of the Erd\"{o}s planar unit-distance problem involved no such iterative human intervention~\cite{9d_OpenAIUnitDistance_26,9d_Sawin_26,9d_Alon_26}. 
The general-purpose reasoning model was not specially trained for that mathematical problem but evaluated across a collection of Erd\"{o}s problems—effectively an APS, rather than being supplied with a problem-specific AHS—and it unexpectedly found a successful route using sophisticated algebraic-number-theoretic ideas.

The present case probes a third, co-development route: through longitudinal Type-III development, the AI not only generated a Type-II hypothesis outside the initial AHS but also autonomously brought a Type-I problem into the shared APS. The AI's subsequent recursive Householder moving-frame construction was a Type-II discovery far beyond the existing AHS, while the human's recognition of its implications qualitatively broadened the scope and deepened the understanding of the AI-generated results with little additional computational effort. Thus, Type-III development can expose a KU whose solution remains a UU, allowing Type III itself to seed new Type-I/II/IV discoveries. The resulting proof synthesizes ordered propagation in 1D transfer constructions and Householder reflection in numerical linear algebra, reminiscent of recent AI-generated discoveries connecting established ideas or tools in unexpected settings~\cite{Gottweis_25_AI_Coscientist,9d_OpenAIUnitDistance_26}.

Taken together, these cases illustrate why autonomous hypothesis generation deserves separate attention from autonomous problem solving and how the two can become connected within a research program. Maximizing single-shot reasoning power and maximizing co-development efficiency are not necessarily the same optimization problem. 

\subsection{From limitations to testable hypotheses\label{sec:being}}

The four discovery types are not proposed as mechanisms inferred from the present case, but as a conceptualization of recurrent modes in the long history of human knowledge acquisition~\cite{Logan_09_Known}. In particular, the role of Type-III development as a cradle for other discovery types has been repeatedly exemplified across scientific inquiry: systematic exploration of an opened landscape can expose longstanding KUs, bring previously unrecognized UUs within reach, and reveal latent connections among existing ingredients. The present case provides a multidimensional human--AI realization of this broader pattern rather than evidence from which the pattern itself is inferred. 
A single case nevertheless cannot establish \emph{quantitative} incidence rates, a necessary sequence of dan levels, or a universal relationship between discovery type and contribution level.
 Testing the developmental hypothesis will require prospective studies across multiple domains, with complete human--AI interaction records, predefined criteria for provenance and novelty, adversarial validation, and comparison against workflows in which AI participation or intervention timing is systematically varied.

Yet, one particularly concrete and potentially testable question arose from this work: the effect of human-intervention timing (being or non-being) on the subsequent discovery trajectory. The present case does not establish that the human's delayed pivot after the AI's rediscovery of the bond-angle--dihedral-angle approach for $n=3$ caused the subsequent arbitrary-\(U,V\) discovery. Unlike most historical counterfactuals, however, the consequences of alternative intervention timing may be experimentally testable by replaying the trajectory with fresh AI instances deprived of subsequent discoveries. For example, without the intervening arbitrary-\(U,V\) result suggesting a geometric origin of the mapping, could an AI proceed directly from the \(n=3\) result—which provides no evident clue toward the Householder construction—to discover the recursive Householder moving-frame method? This suggests controlled comparisons of immediate and delayed human intervention, complementing ongoing research on intervention timing in human--AI workflows~\cite{Zhou_26_AI_When}. 
Such counterfactual trajectories can be represented as alternative paths through the four-axis AHS in Fig.~\ref{fig:AHS}, extending the role of representation from the physical problem to the discovery process itself.

\section{Summary\label{sec:summary}}

In collaboration with a general-purpose AI, we proved a theorem exactly mapping, at the Hamiltonian level, a zero-field range-two \(O(n)\) open chain with arbitrary inhomogeneous interaction functions \(U_i(\bm S_i\!\cdot\!\bm S_{i+1})\) and \(V_i(\bm S_i\!\cdot\!\bm S_{i+2})\) onto a range-one \(O(n)\) chain with NN interaction \(V_i(\boldsymbol\sigma_i\!\cdot\!\boldsymbol\sigma_{i+1})\) and axial single-spin potential \(U_i(\sigma_i^z)\), for every \(n\ge1\) and every system size $L\ge 1$. An analogous theorem for the \(q\)-state Potts model yields a closed-form exact solution of the \(J_1\)-\(J_2\) Potts open chain for every \(q\ge2\) and finite size.

Starting from a human-designed one-axis Type-III AHS for the homogeneous linear thermodynamic-limit specification, the AI autonomously generated hypotheses and microscopic proofs far beyond that AHS, culminating in a recursive Householder moving-frame proof for arbitrary interaction functions. The understanding of this proof in terms of reciprocity and the final extension to arbitrary inhomogeneity were formulated by the human, highlighting the critical role of human recognition and understanding of the implications of the AI-generated results. The trajectory demonstrates AI hypothesis generation outside the human collaborator's AHS and even APS, providing a proof of concept that sustained AI participation in Type-III development may incubate higher-autonomy Type-I/II/IV scientific breakthroughs. It further suggests that alternative intervention paths through the resulting four-axis AHS may make aspects of the discovery process itself experimentally testable.

\begin{acknowledgments}
The author is grateful for helpful discussions with Arthur Ramirez, Gang Cao, and Mark Dean on frustrated magnets, G. Baskaran on the $J_1$-$J_2$ Ising chain, Zvonko Glumac on the standard Potts chain, Alexei Tsvelik on the Heisenberg chain and the \emph{Daodejing}, Peggy Yin and Fang Dong on growth mindset, discovery types,  cognitive security, and the \emph{Daodejing}, Yabin Chen on effective strong-
vs-weak 
mapping, Kevin Wang on representation compatibility, and Felix Archampong, Jianming Bai, Yangang Liu, Alexandre Sitnikov, and Joseph Woicik on experimenting with an effective highly frustrated spin system. The author thanks Daniel Love and Aaron Wilkowitz for their assistance in exporting the human--AI conversation transcript to JSON files.  Brookhaven National Laboratory was supported by the U.S. Department of Energy (DOE), Office of Basic Energy Sciences (BES), Division of Materials Sciences and Engineering under Contract No. DE-SC0012704. 

The AI collaborator used the OpenAI GPT-5.6 family of general-purpose reasoning models. GPT-5.6 Sol in High mode was used for scientific discovery (4d, 6d, and 9d), and Extra High mode was used for critical manuscript review (2d), while GPT-5.6 Sol Pro was used for a single final round of adversarial manuscript review before submission (1d). The human author reviewed and verified all AI-generated scientific content, prepared all figures and tables, and wrote the manuscript with input from the AI. The human author assumes full responsibility for the manuscript.

Wolfram Mathematica 14.3 was used to design the initial Type-III AHS and verify the AI-generated results.
\end{acknowledgments}

\section*{Data availability}
The timestamped human--AI interaction records supporting the discovery-provenance analysis, together with other supporting data, are openly available~\cite{data:9d}.


\appendix
\section{Low-dimensional special cases and symmetry structure\label{appendix:n123}}

The general Householder construction becomes especially transparent for
$n=1,2,3$.  Besides providing direct checks of Theorem~\ref{thm:UiVi}, these cases
clarify the symmetry structure hidden in the microscopic mapping.  The three cases below show how
this general statement reduces to familiar low-dimensional structures.

\subsection{$n=1$: Ising endpoint\label{appendix:special_n1}}

For \(n=1\), the unit sphere is a two-value set $S^0=\{+1,-1\}$ and the \(n\times n\) matrices reduce to scalars. Therefore, the Householder reflector and the moving-frame matrix are simply elements of \(O(1)=\{+1,-1\}\).
With $\hat{e}=1$, the Householder reflector given by Eq.~(\ref{eq:householder_section}) is
\begin{equation*}
\mathsf H(\sigma_i)=\sigma_i.    
\end{equation*}
Here $\mathsf H(\hat{\bm e})=1$ is the conventional choice. Then, \ignore{
\begin{eqnarray*}
 \mathsf G_1=\mathsf H(S_1)=S_1, &
\sigma_1=\mathsf G_1^{\mathsf T}S_2=S_1S_2,\\
\mathsf G_2=\mathsf G_1\mathsf H(\sigma_1)=S_1S_1S_2=S_2, &
\sigma_2=\mathsf G_2^{\mathsf T}S_3=S_2S_3,\\
\mathsf G_3=\mathsf G_2\mathsf H(\sigma_2)=S_2S_2S_3=S_3, & \sigma_3=\mathsf G_3^{\mathsf T}S_4=S_3S_4, \\
\cdots, & \cdots \label{eq:n1_HG}
\end{eqnarray*}
so that} $S_i=\mathsf G_i$ and $\sigma_i=\mathsf G_i^{\mathsf T}S_{i+1}=S_iS_{i+1}$.
Hence
\begin{equation}
\sigma_i\sigma_{i+1}=(S_iS_{i+1})(S_{i+1}S_{i+2})=S_iS_{i+2},
 \label{eq:n1_identities}
\end{equation}
because $S_{i+1}^2=1$, as in the classical
bond-variable construction~\cite{Dobson_JMathP_69_Many-Neighbored-Ising-Chain}.
The utility of reflection is almost invisible.

\subsection{$n=2$: planar and clock chains\label{appendix:special_n2}}

Define the planar rotation $R(\alpha)$ and the reflection across the $x$ axis $D$,
\begin{equation}
 R(\alpha)=
 \begin{pmatrix}
  \cos\alpha&-\sin\alpha\\
  \sin\alpha& \cos\alpha
 \end{pmatrix},
 \qquad
 D=
\begin{pmatrix}1&0\\0&-1\end{pmatrix}.
 \label{eq:n2_RD}
\end{equation}
With $\hat{\bm e}=(1,0)^{\mathsf T}$---the reference unit vector conventionally chosen along the \(x\)-axis in two dimensions---the spin vectors are
\begin{eqnarray}
 \bm S_i&=&
 \begin{pmatrix}\cos\Theta_i\\ \sin\Theta_i\end{pmatrix}=R(\Theta_i)\hat{\bm e},
 \\
 \boldsymbol\sigma_i&=&
 \begin{pmatrix}\cos\phi_i\\ \sin\phi_i\end{pmatrix}=R(\phi_i)\hat{\bm e}.
 \label{eq:n2_spins}
\end{eqnarray}
The Householder reflector given by Eq.~(\ref{eq:householder_section}) is 
\begin{equation}
\mathsf H(\boldsymbol\sigma_i)
  =
 \begin{pmatrix}
  \cos\phi_i& \sin\phi_i\\
  \sin\phi_i&-\cos\phi_i
 \end{pmatrix}
 =R(\phi_i)D,
 \label{eq:n2_H}
\end{equation}
which indeed satisfies Eq.~(\ref{eq:householder_properties}):
$$\mathsf H(\boldsymbol\sigma_i)^{\mathsf T}=\mathsf H(\boldsymbol\sigma_i), \qquad\mathsf H(\boldsymbol\sigma_i)^2=\mathsf I,\qquad \mathsf H(\boldsymbol\sigma_i)\hat{\bm e}=\boldsymbol\sigma_i,$$ because $D^{\mathsf T}=D$, $D^2=\mathsf I$, $D\,\hat{\bm e}=\hat{\bm e}$, $R(\phi_i)^{\mathsf T}=R(-\phi_i)$, and $D R(-\phi_i)=R(\phi_i)D$. Here one may choose $\mathsf H(\hat{\bm e})=D$, instead of $\mathsf H(\hat{\bm e})=\mathsf I$, for  continuously approaching the singular point.

Then,
\ignore{\begin{align*}
 \mathsf G_1&=\mathsf H(\bm S_1)=R(\Theta_1)D, \\
\boldsymbol\sigma_1&=\mathsf G_1^{\mathsf T}\bm S_2=DR(-\Theta_1)\bm S_2=R(\Theta_1-\Theta_2)\hat{\bm e}=R(\phi_1)\hat{\bm e},\\
\mathsf G_2&=\mathsf G_1\mathsf H(\boldsymbol\sigma_1)=R(\Theta_1)D\;R(\Theta_1-\Theta_2)D=R(\Theta_2), \\
\boldsymbol\sigma_2&=\mathsf G_2^{\mathsf T}\bm S_3=R(\Theta_3-\Theta_2)\hat{\bm e}=R(\phi_2)\hat{\bm e},\\
\mathsf G_3&=\mathsf G_2\mathsf H(\boldsymbol\sigma_2)=R(\Theta_2)\;R(\Theta_3-\Theta_2)D=R(\Theta_3)D, \\
\boldsymbol\sigma_3&=\mathsf G_3^{\mathsf T}\bm S_4=DR(-\Theta_3)\bm S_4=R(\Theta_3-\Theta_4)\hat{\bm e}=R(\phi_3)\hat{\bm e}, \\
& \cdots \label{eq:n1_HG}
\end{align*}
Hence,} the recursive moving frame has the explicit form
\begin{equation}
 \mathsf G_i=R(\Theta_i)D^{i},
 \label{eq:n2_G}
\end{equation}
satisfying $\mathsf G_i\hat{\bm e}=R(\Theta_i)\hat{\bm e}=\bm S_i$, while 
\begin{align}
  \boldsymbol\sigma_i&=\mathsf G_i^{\mathsf T}\bm S_{i+1}=D^iR(-\Theta_i)R(\Theta_{i+1})\hat{\bm e}\nonumber \\
  &=R\Big((-1)^i(\Theta_{i+1}-\Theta_i)\Big)\hat{\bm e}.
  \label{eq:n2spin}
\end{align}
Matching Eqs.~(\ref{eq:n2_spins}) and (\ref{eq:n2spin})
leads to the bond-angle definition, 
\begin{equation}
 \phi_i=(-1)^{i}(\Theta_{i+1}-\Theta_i)
 \pmod{2\pi}.
 \label{eq:n2_staggered_bond}
\end{equation}
Thus the staggered bond angle is the explicit $n=2$ form of the alternating orientation of the recursive reflection frame.  Indeed,
\begin{align}
 \bm S_i\!\cdot\!\bm S_{i+1}
 &=\cos(\Theta_{i+1}-\Theta_i)=\cos\phi_i=\hat{\bm e}\cdot \boldsymbol\sigma_i,
 \label{eq:n2_nn}\\
 \bm S_i\!\cdot\!\bm S_{i+2}
 &=\cos(\Theta_{i+2}-\Theta_i) \nonumber\\
&=\cos(\phi_{i+1}-\phi_{i})=\boldsymbol\sigma_i\!\cdot\!\boldsymbol\sigma_{i+1}.
 \label{eq:n2_nnn}
\end{align}

For the planar model the full global symmetry of a dot-product Hamiltonian is
\begin{equation}
 O(2)=SO(2)\rtimes\mathbb Z_2
 \simeq U(1)\rtimes\mathbb Z_2.
 \label{eq:O2_symmetry}
\end{equation}
The bond angle is naturally coordinatized by the Abelian rotation subgroup $SO(2)\simeq U(1)$, whereas the transformation $\phi\mapsto-\phi$ is the reflection component.  Equations (\ref{eq:n2spin}) and
\eqref{eq:n2_staggered_bond} therefore show how reflection is built into the planar Hamiltonian-level mapping through alternating bond inversion. Once revealed, this is recognized as the low-dimensional precursor of the Householder reflection used for arbitrary $n$.

\textbf{For the $q$-state clock model}, which is a discrete $n = 2$ model with
$\Theta_i=2\pi k_i/q$ where $k_i\in\mathbb Z_q$, 
Eq.~\eqref{eq:n2_staggered_bond} closes within the same discrete set.  The full polygon symmetry is the dihedral group
\begin{equation}
 D_q=\mathbb Z_q\rtimes\mathbb Z_2,
 \label{eq:Dq_symmetry}
\end{equation}
where the $\mathbb Z_q$ factor gives cyclic rotations and the $\mathbb Z_2$ factor acts by inversion $k\mapsto-k$.  The clock mapping is therefore the immediate discrete restriction of the planar mapping, rather than a separate geometric construction.  In particular, the staggered bond variable uses precisely the inversion automorphism of $\mathbb Z_q$.

\subsection{$n=3$: Heisenberg chain\label{appendix:special_n3}}

Define the proper rotations about the $z$ and $y$ axes, respectively,
\begin{align}
 R_z(\alpha)&=
 \begin{pmatrix}
  \cos\alpha&-\sin\alpha&0\\
  \sin\alpha& \cos\alpha&0\\
  0&0&1
 \end{pmatrix},
 \label{eq:n3_Rz}\\
 R_y(\theta)&=
 \begin{pmatrix}
  \cos\theta&0&\sin\theta\\
  0&1&0\\
  -\sin\theta&0&\cos\theta
 \end{pmatrix},
 \label{eq:n3_Ry}
\end{align}
and the fixed reflection across the $yz$ plane,
\begin{equation}
 D=\operatorname{diag}(-1,1,1),
 \qquad \det D=-1.
 \label{eq:n3_D}
\end{equation}
With $\hat{\bm e}=(0,0,1)^{\mathsf T}$, the spin vectors in spherical coordinates are
\begin{eqnarray}
 \bm S_i=
 \begin{pmatrix}
  \sin\Theta_i\cos\Phi_i\\
  \sin\Theta_i\sin\Phi_i\\
  \cos\Theta_i
 \end{pmatrix}=R_z(\Phi_i)R_y(\Theta_i)\hat{\bm e}
 \\
 \boldsymbol\sigma_i=
 \begin{pmatrix}
  \sin\theta_i\cos\alpha_i\\
  \sin\theta_i\sin\alpha_i\\
  \cos\theta_i
 \end{pmatrix}=R_z(\alpha_i)R_y(\theta_i)\hat{\bm e},
 \label{eq:n3_sigma}
\end{eqnarray}

The Householder reflector given by Eq.~(\ref{eq:householder_section}) is
\begin{eqnarray}  
 \mathsf H(\boldsymbol\sigma_i)&=&
 \begin{pmatrix}
 s_{\alpha_i}^{2}-c_i c_{\alpha_i}^{2}
 &-(1+c_i)c_{\alpha_i}s_{\alpha_i}
 &s_i c_{\alpha_i}\\
 -(1+c_i)c_{\alpha_i}s_{\alpha_i}
 &c_{\alpha_i}^{2}-c_i s_{\alpha_i}^{2}
 &s_i s_{\alpha_i}\\
 s_i c_{\alpha_i}
 &s_i s_{\alpha_i}
 &c_i
 \end{pmatrix} \label{eq:n3_H_explicit} 
 \nonumber
\end{eqnarray}
where $c_i=\cos\theta_i$, $s_i=\sin\theta_i$,
$c_{\alpha_i}=\cos\alpha_i$, and
$s_{\alpha_i}=\sin\alpha_i$. 
It admits a particularly transparent factorization,
\begin{eqnarray}  
 \mathsf H(\boldsymbol\sigma_i)&=&
 R_z(\alpha_i)R_y(\theta_i)D R_z(-\alpha_i),
 \label{eq:n3_H_factorized}
\end{eqnarray}
which makes the symmetry content explicit: a local
Householder update is a pair of proper rotations surrounding a fixed reflection. Equivalently, it is the reflection $D$ carried to the appropriate local plane. 
For fixed $\Phi_i=\alpha_i= 0$, the $n=3$ problem reduces to the $n=2$ one in the $zx$ plane, and indeed Eq.~\eqref{eq:n3_H_factorized} reduces to Eq.~(\ref{eq:n2_H}) upon the $zx\to xy$ transformation. 

The two microscopic invariants (\ref{eq:geom_nn_identity2}) and (\ref{eq:geom_nnn_identity2}) now read
\begin{eqnarray}
\bm S_i\!\cdot\!\bm S_{i+1}
 &=&\hat{\bm e}\!\cdot\!\boldsymbol\sigma_i
 =\cos\theta_i.
 \label{eq:n3_nn2}\\
 \bm S_i\!\cdot\!\bm S_{i+2}
 &=&\boldsymbol\sigma_i\cdot\boldsymbol\sigma_{i+1}
 \nonumber\\
 &=&c_ic_{i+1}
   +s_is_{i+1}\cos(\alpha_{i+1}-\alpha_i),
 \label{eq:n3_nnn} 
\end{eqnarray}
reproducing Eqs.~(\ref{eq:XY_sigmaz})--(\ref{eq:XY_sigma}) after matching the bond- and dihedral-angles: $\vartheta_i=\theta_i$
, $\vartheta_{i+1}=\theta_{i+1}
$, and $\varphi=\alpha_{i+1}-\alpha_i$, up to the orientation/sign convention for the dihedral angle, which is irrelevant inside the cosine.


%

\end{document}